\documentclass[onecolumn]{revtex4-2}
 
\usepackage{amsmath,dsfont, mathrsfs, amsfonts,amssymb, amsthm, bm}
 
\usepackage{tikz}
\usetikzlibrary{quantikz}

\usepackage[colorlinks=true,
allcolors=blue]{hyperref}
\usepackage{braket}
\usepackage{algorithm} 
\usepackage[noend]{algpseudocode}

\DeclareMathAlphabet{\mymathbb}{U}{BOONDOX-ds}{m}{n}

\newcommand{\soft}{\widetilde}
\newcommand\poly{\mathrm{poly}} 
\newcommand\polylog{\mathrm{polylog}} 
\newcommand{\con}{\text{c-}}
\newcommand{\Con }{\text{c-}}
\newcommand{\norm}[1]{\lVert #1 \rVert}
\newcommand{\abs}[1]{\left | #1\right |}

\theoremstyle{definition}
\newtheorem{thm}{Theorem}[section]
\newtheorem{lemma}[thm]{Lemma}
\newtheorem{coro}[thm]{Corollary}

\newtheorem{problem}[thm]{Problem}

\newtheorem{definition}[thm]{Definition}

\usepackage[normalem]{ulem}
\begin{document}

\title{
Tight Bound for Estimating Expectation Values
from a System of Linear Equations
}

\date{\today}

\author{Abhijeet Alase}
\author{Robert R. Nerem}
\author{Mohsen Bagherimehrab}
\author{Peter H\o{}yer}
\author{Barry C.\ Sanders}
\affiliation{Institute for Quantum Science and Technology, University of Calgary, 2500 University Drive NW, Calgary, Alberta T2N 1N4, Canada}

\begin{abstract}
The System of Linear Equations  Problem (SLEP) is specified by a complex invertible matrix~$A$, 
the condition number $\kappa$ of $A$,  a vector~$\bm{b}$, a Hermitian matrix $M$ and an accuracy $\epsilon$,
and the task is to estimate $\bm{x}^\dagger M \bm{x}$,
where $\bm{x}$ is the solution vector to the equation~$A\bm{x} = \bm{b}$.
We aim to establish a lower bound on the complexity of the  
end-to-end quantum algorithms for SLEP with respect to~$\epsilon$, 
and  devise a  quantum algorithm that saturates this bound. 
To make lower bounds attainable,  we consider query complexity in the setting in which a block encoding of $M$ is given, i.e.,  a unitary black box $U_M$ that contains $M/\alpha$ as a block for some $\alpha \in \mathbb R^+$.
We show{, by constructing a quantum algorithm and deriving a lower bound,}  that the quantum query complexity for SLEP in this setting
is $\Theta(\alpha/\epsilon)$.
Our lower bound is established by reducing the problem of estimating the mean of a 
black box function to SLEP.
Our~$\Theta(\alpha/\epsilon)$ result tightens and proves the common assertion 
of polynomial accuracy dependence ($\poly(1/\epsilon)$) for SLEP { without making any complexity-theoretic assumptions}, 
and shows that improvement beyond linear dependence 
on accuracy is not possible if $M$ is provided via block encoding. 
\end{abstract}

\maketitle


\section{Introduction}
\label{sec:intro}
Systems of linear equations are essential to nearly every area of science, engineering and mathematics, being particularly valuable in numerical solutions to differential equations and machine-learning applications~\cite{WBL12,Berry2014,LMR2014,RML2014,BWPRWL2017}.
For the System of Linear Equations Problem (SLEP), a size~$N$,
an $N \times N$ complex invertible matrix~$A$, an upper bound $\kappa$ on the condition number  of $A$, a vector~$\bm{b}$, a Hermitian matrix $M$ and an accuracy $\epsilon$ are given,
and the task is to estimate $\bm{x}^\dagger M \bm{x}$,
where $\bm{x}$ is the complex solution vector to the equation~$A\bm{x} = \bm{b}$.
If the number of equations~$N$ is large, such as for big-data applications, solving SLEP is intractable for the best-known classical algorithms, which require polynomial resources in $N$. The landmark quantum algorithm by Harrow, Hassidim and Lloyd (HHL algorithm)~\cite{HHL09}, together with its improvements~\cite{Amb12,CKS17,DYY2020}, generate, to accuracy $\epsilon$, a quantum state $\ket {\bm x}$  encoding the solution $\bm x$
using time and space resources that are polylogarithmic in $N$ and~$1/\epsilon$, which is denoted by $\polylog(N,1/\epsilon)$. SLEP can be solved by estimating 
the expectation value $\bm{x}^\dagger M \bm{x}$ using such quantum algorithms for generating $\ket{\bm{x}}$.  Although it is commonly asserted that estimating $\bm x ^\dagger M \bm x$ requires $\poly(1/\epsilon)$ resources~\cite{ L2021,Gilyn2019, Chia2020}, 
this assertion is proven only under the assumption that $A$ is given by oracles~\cite{HHL09}.
To this end, we establish a lower bound for the complexity of SLEP with respect to $\epsilon$ if $M$ (instead of $A$) is given by an oracle, and we saturate this bound by constructing an end-to-end quantum algorithm for SLEP.

We now summarize the state-of-the-art for the hardness of SLEP in terms of~$\epsilon$.
Hardness results were proved for a restricted version of SLEP in which $b$ is fixed to be a vector with a $1$ in its first entry and zeros elsewhere, and $M$ is fixed to be the diagonal matrix containing a 1 in the first $N/2 = 2^{n-1}$ diagonal entries and zeros in all other diagonal entries~\cite{HHL09}.
We refer to this restriction of SLEP with fixed $M$ and $b$ as F-SLEP ({ See Table \ref{tab:problems} 
for a list of problems considered in this paper}).
Under a reasonable complexity-theoretic assumption ($\textbf{BQP} \neq \textbf{PP}$), a $\polylog(1/\epsilon)$ dependence was ruled out
for any quantum algorithm for F-SLEP that achieves $\poly(\log N,\kappa)$ dependence.
Furthermore, any quantum algorithm solving F-SLEP with $\poly(\log N,\kappa)$ dependence cannot achieve runtime faster than $\mathcal O(1/\epsilon)$ if $A$ is provided by an oracle. As F-SLEP is a restriction of SLEP, these hardness results extend to SLEP as well. However, complexity lower bounds that do not rely on complexity-theoretic assumptions are not known in the case that $A$ is fixed. Moreover, whether a quantum algorithm with subpolynomial accuracy dependence can solve SLEP is not known in this case. 

We next discuss previously developed end-to-end quantum algorithms for SLEP, which mirror the literature on SLEP lower bounds in their focus on restricted cases. For F-SLEP, $M$ is the matrix representation of the observable $\ket{0}\bra{0}\otimes\mathds{1}_{n-1}$
in the computational basis.
Therefore, the expectation value
$\bm{x}^\dagger M \bm{x}$ can be computed by repeatedly preparing the solution state $\ket{\bm{x}}$ and measuring the first qubit in
the computational basis, then recording the fraction of measurements that yield a $0$ outcome~\cite{HHL09}.
Efficient quantum algorithms were developed for  
estimating $|x_i|^2$ or $|x_i-x_j|^2$ for some indices~$i,j$, with $x_i$ the $i$th entry of $\bm x$~\cite{Wan17}.
These cases, which are solved using amplitude estimation, are  restrictions of SLEP to instances where entries of $M$ are zero except for a single~1 on the diagonal, or except for a single~1 and a single $-1$ on the diagonal, respectively.
In a quantum algorithm developed for data fitting, fit accuracy is computed by solving SLEP in the case $M = A$, i.e., the case that $\braket{\bm{x}|M|\bm{x}} \propto \abs{\bm{b}\cdot\bm{x}}$~\cite{WBL12}, where the expectation value is computed by the swap test~\cite{BCW2001}.   
As these algorithms apply only to a few restrictions of SLEP, computing expectation values in general could still be 
prohibitively expensive~\cite{Aaronson2015}; therefore, unclear complexities for general expectation-value estimation obfuscate potential applications of quantum algorithms for SLEP.

We now discuss our claims and the setting in which they are developed. We construct our lower bound and quantum algorithm in the block-encoding framework which provides an alternative over the sparse-access model in which matrices are provided by oracles for locating and evaluating their non-zero entries~\cite{ATS2003, BAC+07, BCCKS2014,CW2012,WBHS2011, BCK15}. A block encoding~$U_H$ of a Hermitian matrix $H$ is a unitary black box $U_H$ that contains $H/\alpha_H $ in the upper-left block for some $\alpha_H \in \mathbb R^+$~\cite{ LC19,CGJ19, Gilyn2019}.
We consider SLEP in the setting in which $M$ is provided by block encoding $U_M$, a version of SLEP we refer to as B-SLEP. We show 
that any quantum algorithm for B-SLEP must make $\Omega(\alpha_M/\epsilon)$ queries to $U_M$. 
We then construct a quantum algorithm that saturates this bound, thereby showing that our lower bound for B-SLEP is tight, meaning that B-SLEP has query complexity $\Theta(\alpha_M/\epsilon)$.
En route to deriving our lower bound for B-SLEP, we evaluate the complexity of estimating the expectation value of a block encoded matrix $M$ with respect to some given quantum state. In particular we show that the query complexity of this problem is also in $\Theta(\alpha_M/\epsilon)$. 
We also derive a query-complexity lower bound for the restriction of SLEP 
to the case in which $A = \mathds{1}_n$ and $M$ is provided via sparse-access. This restriction is equivalent to the problem of estimating the expectation value of a matrix $M$ in the sparse-access setting. 
Our lower bound for this problem is $\Omega(d_M \beta_M/\epsilon)$, 
where the sparsity~$d_M$ of $M$ and a bound~$\beta_M$ on the max-norm of $M$ 
are supplied as inputs to the problem. 
We show that this lower bound is tight by constructing an explicit quantum algorithm.

Methodologically, we utilize the block-encoding framework to limit restrictions on $M$ and to make lower bounds attainable. The key step for deriving our lower bound is reducing the problem of computing the mean of a black-box function to the problem of computing an expectation value of a block-encoded Hermitian matrix (B-EVHM). 
This lower bound is extended to B-SLEP by showing that B-EVHM reduces to B-SLEP. 
Our algorithm saturating this lower bound is constructed by combining a known algorithm for
B-EVHM~\cite{Ral2020} with techniques for generating $\ket{\bm{x}}$.

Our complexity bound for B-SLEP makes rigorous the common assertion of $\poly(1/\epsilon)$ dependence for SLEP and shows that improvement beyond linear-accuracy dependence is impossible if $M$ is given by oracles. 
Our end-to-end quantum algorithm for solving systems of linear equations allows for a variable $M$ as input, whereas previous algorithms consider a restricted $M$. Explicit complexities for SLEP are important to assessing if the measurement step is prohibitively resource-expensive for a given application, thereby addressing the limitation raised by  Aaronson's fourth caveat for quantum linear-equation solvers~\cite{Aaronson2015}. 

An outline of our work is as follows. In \S\ref{sec:background} we give background relevant to our results including a discussion of matrix encodings, quantum algorithms for linear algebra, and  quantum algorithms for estimating expectation values. Next, in \S\ref{sec:approach}, we describe our methodology and define B-SLEP rigorously. We present our lower bounds and quantum algorithms in \S\ref{sec:results}  and discuss these results in  \S\ref{sec:discussion}. Finally, we give our conclusions in  \S\ref{sec:conclusion}.

\begin{table}[t]
\label{tab:problems}
{
\begin{tabular}{|l|l|l|l|}
\hline
\textbf{Abbreviation} & \textbf{Problem Name}                                          & \textbf{Number}         & \textbf{Comment}           \\ \hline
B-SLEP       & Block-access System of Linear Equations Problem       & \ref{pro:SLE1}   & SLEP with block-access to $A$,$M$\\ \hline
F-SLEP       & Fixed System of Linear Equations Problem              & \ref{prob:FSLEP} & SLEP in sparse-access setting\\ 
&&& with the restriction $M=\ket{0}\bra{0}\otimes\mathds{1}_{n-1}$                                 \\ \hline
S-EVHM       & Sparse-access Expectation Value of a Hermitian Matrix & \ref{pro:S-EVHM} & Expectation value of a sparse $M$ \\ \hline
B-EVHM       & Block-access Expectation Value of a Hermitian Matrix  & \ref{pro:B-EVHM} & Expectation value of a block-encoded $M$\\ \hline
S-QLSP       & Sparse-access Quantum Linear System Problem   & \ref{prob:SQLSP} & Output is an approximation of $\ket{\bm{x}}$ \\ \hline
\end{tabular}
\caption{Abbreviation, name and defining feature of all restrictions of SLEP and relevant computational problems considered in this paper.}}
\end{table}

\section{Background}
\label{sec:background}
In this section, we first review the quantum linear systems problem, which is the computational problem of generating
the solution state $\ket{\bm x}$ given $A$, $\bm b$, and $\kappa$, in the sparse-access setting. 
We then review the block-encoding framework of Hermitian matrices and relevant techniques for matrix arithmetic
in this framework. Next, we state results relevant for deriving  lower bound on accuracy dependence for the system of linear equations problem. Finally, we review
known quantum algorithms for computing the expectation value of an unknown state in the context of system of linear equations or
otherwise.

We start with some preliminaries. Denote the set of integers and positive integers by $\mathbb{Z}$ and $\mathbb{Z}^+$ respectively,
and for $N \in \mathbb Z^+$ define
\begin{equation}
    [N] := \{0,1,\dots,N-1 \}.
\end{equation}
Similarly, our convention for indexing rows and columns of a matrix is zero-based. 
We denote the real and complex fields by $\mathbb{R}$ and $\mathbb{C}$ respectively, and the set of positive real numbers by $\mathbb{R}^+$.  A number $ \tilde z$ is an $\epsilon$-additive approximation to another number $z$ if $|z-\tilde z| < \epsilon$.  Although we do not define a specific bit-string representation for real numbers and complex numbers, our results hold for any representation in which a number $z$ has an $r$-bit representation encoding a number $z'$ such that 
$|z-z'|/|z| \in \mathcal O(2^{-r})$, and basic arithmetic operations on representations can be performed to 
$r$ bits of precision in $\poly(r)$ time. 
These condition holds for standard 
representations of real and complex numbers~\cite{Brent1976}. For a vector $\ket{\psi} \in \mathscr{H}$ 
in a finite-dimensional Hilbert space
$\mathscr{H}$, we define $\|\ket{\psi}\| := \sqrt{\braket{\psi |\psi}}$, and for operators $V: \mathscr H \to \mathscr H$ define $\|V\|$ to be the corresponding operator norm. For a matrix $M \in \mathbb C^{N \times N}$, denote $\|M\|_{\mathrm{max}} := \max_{i,j \in [N]}\{|M_{ij}|\}$ as the max norm.  For any $n \in \mathbb Z^+$ we define $\mathds 1_n$ to be the $n$-qubit identity operator. If $n$ is clear from context, we may omit the subscript.

\subsection{Quantum linear system problem in the sparse-access setting}
\label{sec:QLSP}
We now explain the sparse-access model for access to Hermitian matrices. Next, we review a formulation 
of the quantum linear systems problem based on this model.

\subsubsection{Sparse access to Hermitian matrices}
\label{sec:sparseaccess}
The quantum linear systems problem~\cite{HHL09,CKS17} accepts an $N\times N$ matrix $A$
and an $N$-dimensional vector  $\bm{b}$ as inputs.
To devise a quantum algorithm for generating the solution state in $\polylog(N)$ time,
 $A$ and $\bm{b}$ must be encoded in such a way that they can be accessed efficiently. Most notably, such an efficiency requirement rules out encoding these inputs into a string containing a list of entries.

The standard formulation of the quantum linear systems problem utilizes unitary black boxes to encode $A$ and $\bm b$~\cite{HHL09, CKS17}. The vector $\bm{b}$ is given
by a black-box  $U_{\bm{b}} \in \mathcal{U}(\mathscr{H}_2^{\otimes n})$ ($n$-qubit unitary operator) such that $U_{\bm{b}}\ket{0^n} = \ket{\bm{b}}$, where $\ket{\bm b} := \sum_i b_i\ket i/\| \sum_i b_i \ket{i}\|$ and $\{\ket{i}\}$ are computational basis states.
The matrix $A$ is described by sparse-access oracles, i.e., oracles for functions that provide the column index
and the value of the non-zero entries of the matrix $A$ in each given row~\cite{ATS2003}. We next explain
in detail the sparse-access setting that we use in this paper. 

For any $n \in \mathbb{Z}^+$, an \textit{$n$-qubit Hermitian matrix} $H$ is a
$2^n \times 2^n$ Hermitian matrix.
An $n$-qubit Hermitian matrix $H$ is \textit{$d_H$-sparse} for $d_H \in [2^n]$ if each row of $H$ has at most $d_H$ non-zero entries. We call $d_H$ the \textit{sparsity parameter} of~$H$. 
In the sparse-access setting, the matrix values $\{H_{j k}\}$ are encoded into both standard unitary oracles~\cite{BCK15}
\begin{equation}
O_{H_{\rm val}}|j\rangle|k\rangle|z\rangle=|j\rangle|k\rangle\left|z \oplus H_{j k}\right\rangle, 
\end{equation}
\begin{equation}
 O_{H_{\rm loc}}|j\rangle|l\rangle=|j\rangle|H_{\rm loc}(j, l)\rangle,
\end{equation}
where the function $H_{\rm loc}$ maps a row index $j$ and a number $l \in[d_H]$ to the column 
index $H_{\rm loc}(j,l)$ of the $l^{\text{th}}$ non-zero element in row $j$. Observe that the
oracle ${O}_{H_{\rm loc}}$ computes the location of entries in-place, i.e.\ in the same quantum register as~$l$.
The oracle ${O}_{H_{\rm val}}$ accepts a row $j$ and column $k$ index and returns the value $H_{j k}$ in some binary format. The pair of oracles ${O}_{H_{\rm val}}$ and ${O}_{H_{\rm loc}}$ are a sparse-access encoding of $H$.  For convenience, we make the following definition. 
\begin{definition}
 \textit{Sparse access} to a Hermitian matrix  $H$ is a 4-tuple $(d_H, \beta_H,O_{H_{\rm val}},O_{H_{\rm loc}})$,
 where  $\beta_H \ge \|H\|_{\max}$ is an upper bound on the max-norm of $H$, $d_H$ is the sparsity parameter
 of $H$, and where $ O_{H_{\rm val}}$ and  $O_{H_{\rm loc}}$ are unitary black boxes that form a sparse-access encoding of $H$.
\end{definition}
\noindent
We frequently turn to this definition in the description of inputs for computational problems in the sparse-access model.

We always assume that, if any unitary black box $O \in \mathcal{U}(\mathscr{H}_2^{\otimes n})$ is given, then unitary black boxes $O^\dagger, \con O,$ and $\con O^\dagger$ are also given where $\con O := \ket{0}\bra{0}\otimes\mathds{1}_n + \ket{1}\bra{1}\otimes O$.  Consequentially, if we are reporting complexities, we equate the cost of a query to $O^\dagger, \con O,$ and $\con O^\dagger$ to the cost of single query to $O$. The motivation for this assumption is that oracles model subroutines that have unknown internal structure and, given a description of a circuit for $O$, circuits for each $O^\dagger, \con O,$ and $\con O^\dagger$ can always be constructed using a procedure independent of the specific circuit for $O$~\cite{AK2007}. In particular, the circuit for  $\con O$ is found by replacing each gate $G$ in the circuit for $O$ with $\con G$. The circuit for $O^\dagger$ is found by reversing the circuit for $O$ and replacing each gate $G$ with its inverse $G^\dagger$. Finally, the circuit for  $\con O^\dagger$ is found by replacing each gate $G$ in the circuit for $O^\dagger$ with $\con G$.

\subsubsection{Complexity and limitations of the HHL algorithm}
We now state the quantum linear system problem rigorously. Next we review the complexity of the HHL algorithm~\cite{HHL09}, which solves the quantum linear systems problem. We then discuss limitations and caveats of this quantum algorithm. 

Given a  system of linear equations 
\begin{equation} \label{eq:Ax = b}
A\bm x = \bm b,
\end{equation} define the solution state $\ket{\bm x}$ to be the quantum state $\ket{\bm x} := \sum_i x_i\ket i/\| \sum_i x_i \ket{i}\|$. The following problem statement is adapted from Ref.~\cite{CKS17}.

\begin{problem}[Sparse-access Quantum Linear System Problem (S-QLSP)]
\label{prob:SQLSP}
Given $n$ qubits, 
a $\kappa \ge 1$,
sparse access $(d_A,\beta_A, O_{A_{\rm val}},O_{A_{\rm loc}})$ to a $2^n\times 2^n$ invertible Hermitian matrix $A$ 
 satisfying $\norm{A} \le 1$ and $\norm{A^{-1}} \le  \kappa$,
an accuracy $\epsilon \in (0,2]$, and an $n$-qubit unitary black box $U_{\bm b}$,
generate with probability at least $2/3$, and return ``failure'' otherwise,
a state $\ket{\tilde{\bm x}}$ such that $\| \ket{\tilde{\bm x}}-\ket{\bm x} \| \le \epsilon$ 
where  $\ket{\bm{x}}$ is the solution state to the system of linear equations $A\bm x = \bm b$ 
 and $\bm{b}$ is the $2^n$-dimensional vector with entries $b_i = \braket{i|U_{\bm{b}}|0^n}$.
\end{problem}
\noindent {Tables~\ref{tab:io} and~\ref{tab:syms} list input and output for the computational problems
discussed in this paper.}

\begin{table}[t]
{
\begin{tabular}{|l|l|l|}
\hline
\textbf{Abbreviation} & \textbf{Input}                                          & \textbf{Output (approximation of)}                    \\
\hline
B-SLEP       & $(n,\kappa,(\alpha_A,a_A,0, U_A),(\alpha_M,a_M,0, U_M),\epsilon,U_{\bm b})$ & $\bm{x}^\dagger M\bm{x}$   \\
\hline
F-SLEP       & $(n,\kappa,(d_A, \beta_A,O_{A_{\rm val}},O_{A_{\rm loc}}),\epsilon,U_{\bm b})$              &        $\bm{x}^\dagger (\ket{0}\bra{0}\otimes\mathds{1}_{n-1})\bm{x}$                           \\
\hline
S-EVHM       & $(n,(d_M, \beta_M, O_{M_{\rm val}},O_{M_{\rm loc}}),\epsilon,V)$ & $\braket{{0^n}|V^\dagger M V|{0^n}}$ \\ \hline
B-EVHM       & $(n,(\alpha_M, a_M,0, U_M),\epsilon,V)$ & $\braket{{0^n}|V^\dagger M V|{0^n}}$ \\ \hline
S-QLSP       & $(n,\kappa,(d_A,\beta_A, O_{A_{\rm val}},O_{A_{\rm loc}}),\epsilon,U_{\bm b})$ & $\ket{\bm{x}}$ \\ \hline
\end{tabular}
\caption{Input and output for the computational problems listed in Table \ref{tab:problems}. \label{tab:io}}}
\end{table}

\begin{table}[t]
{
\begin{tabular}{|l|l|}
\hline
\textbf{Symbol} & \textbf{Description}  \\ \hline
$n$ & Number of qubits \\ \hline
$\kappa$ & Upper bound on the condition number \\ \hline
$A$, $M$ & Hermitian matrices of size $2^n \times 2^n$\\ \hline
$\bm{x}$, $\bm{b}$ & Vectors of length $2^n$ \\ \hline
$\alpha_A$ ($\alpha_M$) & Scaling factor for the block encoding of $A$ ($M$)\\ \hline
$a_A$ ($a_A$) & Ancilla qubits for block encoding of $A$ ($M$)\\ \hline
$U_A$ ($U_M$) & Unitary block encoding of $A$ ($M$)\\ \hline
$d_A$ ($d_M$) & Sparsity of $A$ ($M$)\\ \hline
$\beta_A$ ($\beta_M$) & Upper bound on the max-norm of $A$ ($M$)\\ \hline
$O_{A_{\rm val}}$ ($O_{M_{\rm val}}$) & Oracle that returns the entry of $A$ ($M$) in a given row and column\\ \hline
$O_{A_{\rm loc}}$ ($O_{M_{\rm loc}}$) & Oracle that returns the non-zero entries of $A$ ($M$) in a given row\\ \hline
$\epsilon$ & Additive error tolerance\\ \hline
$U_{\bm b}$ & Arbitrary black-box unitary. The vector $\bm{b}$ is generated by $U_{\bm b}$\\ \hline
$V$ & Arbitrary black-box unitary\\ \hline
\end{tabular}
\caption{Description of symbols used in Table \ref{tab:io}.\label{tab:syms}}}
\end{table}

The crucial step in the HHL algorithm is a call to the phase-estimation algorithm, which
in turn makes use of a quantum algorithm for sparse Hamiltonian simulation. The complexity of the HHL algorithm
is $\poly(\log N, \kappa, 1/\epsilon)$, where the $\poly(1/\epsilon)$ dependence comes from the phase-estimation step. This complexity has been improved to $\poly(\log N, \kappa, \log(1/\epsilon))$
by replacing phase estimation with quantum-walk techniques for implementing $A^{-1}$~\cite{CKS17}.
These techniques have been generalized to the larger class of block-encoded matrices~\cite{LC19},
which we discuss in the next section.

The HHL algorithm has four key limitations on its ability to solve linear equations in practice~\cite{Aaronson2015}. First, this algorithm only maintains an exponential speedup over classical techniques for instances of S-QLSP  for which $\kappa \in \mathcal O(\log(N))$. Next, the algorithm only applies to systems of linear equations in which the oracle $U_{\bm b}$ can be efficiently constructed. Similarly, the algorithm is only applicable if $A$ can be encoded by sparse-access oracles, and the speedup is lost if $d_A \notin \polylog(N)$. 
Finally, the output of HHL is a quantum state $\ket {\bm x}$ encoding the solution $\bm x$. As a result, the HHL algorithm is only useful for finding quantities that can be efficiently  computed using $\ket {\bm x}$. The purpose of our work is to address this final caveat by evaluating the complexity of estimating the quantities of the form $\bm x^\dagger M \bm x$ given $M$.

\subsection{Block-encoding techniques for the quantum linear system problem }
\label{sec:backgroundqwalk}

Quantum-walk techniques for implementing $H^{-1}$ and Hamiltonian simulation (implementing 
$\text{e}^{\text{i}H}$) are applicable not only to sparse matrices, but to all matrices
that can be encoded efficiently as blocks of larger unitary matrices.
This realization has fueled development of quantum algorithms for linear algebra
in which input matrices are provided through a block encoding~\cite{LC19,CGJ19, Gilyn2019}. We also derive our results in the block-encoding setting. In this section, we review block encoding
of Hermitian matrices and describe techniques for matrix arithmetic that we use in the derivation of our results.

\subsubsection{Block encoding for Hermitian matrices}
\label{sec:block encodings}
Block encoding is a way of specifying a matrix input to a computational problem, and therefore provides an
alternative to sparse-access encoding reviewed in \S\ref{sec:sparseaccess}. 
Intuitively, a block encoding of a Hermitian matrix $H$ is a unitary matrix $U_H$ which contains $H$ as an upper-left block.
We now adapt a rigorous definition of block encoding~\cite{LC19}.
\begin{definition}\label{def:block encoding}
Let $H$ be an $n$-qubit Hermitian matrix, $\alpha_H, \delta_H > 0 $ and $a_H \in \mathbb{Z}^+$. 
A $U_H  \in \mathcal{U}(\mathscr{H}_2^{\otimes (n+a_H)})$ is an $(\alpha_H,a_H, \delta_H)$-block encoding of $H$ if 
\begin{equation}
    \norm{ H  - \alpha_H \bra{0^{a_H}}   U_H \ket{0^{a_H}} }\leq \delta_H,
\end{equation}
where $ \bra{0^{a_H}}   U_H \ket{0^{a_H}}  := \sum_{i,j \in [2^n]}\left(\bra{0^{a_H}}\bra{i} U_H \ket{0^{a_H}}\ket{j} \right) \ket{i}\bra{j}$.
\end{definition}

As $\alpha_H,a_H$ and $ \delta_H$ are  provided along with $U_H$ as inputs to a problem, we use the following convention. 
\begin{definition}
\textit{Block access} to a Hermitian matrix $H$ is a tuple $(\alpha_H,a_H, \delta_H, U_H)$, where $U_H$ is a unitary black box that is a $(\alpha_H,a_H, \delta_H)$ block encoding of $H$.
\end{definition}
\noindent 
If block access $(\alpha_H,a_H, \delta_H, U_H)$ to $H$ is given, we assume that
$U_H^\dagger$, $\con U_H$ and $\con U_H^\dagger$ are also given.

If a Hermitian matrix is given by sparse-access oracles, then a previously-developed algorithm can be used to construct
a block-encoding for the same Hermitian matrix~\cite{LC19}. 
We now state a previously-known lemma~\cite{LC19} that gives the query complexity of obtaining a block encoding 
from a sparse-access encoding. We adapt this lemma to include the complexity in terms of 2-qubit gates.
\begin{lemma}\label{lem:block encoding sparse access}
Given a number of qubits $n \in \mathbb{Z}^+$ and sparse access  $(d_H,\beta_H, O_{H_{\rm val}},O_{H_{\rm loc}})$ to an $n$-qubit Hermitian matrix $H$, 
 block access $\left(d_H\beta_H, 2, \epsilon, U_H\right)$ to $H$ can be constructed such that $U_H$  can be implemented 
with $\mathcal{O}(1)$ queries to sparse-access oracles for $H$ and 
$\mathcal{O}\left(n+\log ^{2.5}\left(\frac{d_H\beta_H}{\epsilon}\right)\right)$ additional 2-qubit gates.
\end{lemma}

\begin{proof}
Define $\{\ket{\psi_{j}}:=\ket{0}\otimes \ket{j} \otimes \ket{\varphi_j} 
\in \mathscr{H}_2 \otimes \mathscr{H}_2^{\otimes n} \otimes \mathscr{H}_2^{\otimes n+1}: j \in[2^n]\}$, where 
\begin{align}
\label{eq:walkstates}
\ket{\varphi_j} :=  \frac{1}{\sqrt{d_H}} \sum_{\substack{k = H_{\rm loc}(j,\ell)\\ \ell \in [d_H]}}
\left(\sqrt{\frac{H_{j k}^{*}}{\beta_H}}\ket{0}\ket{k}+
\sqrt{1-\frac{\left|H_{j k}\right|}{\beta_H}}\ket{1}\ket{k}\right)
\quad
\forall\, j \in [2^n],
\end{align}
with $\sqrt{\bullet}$ denoting the principal square root \footnote{Note that special care must be taken 
in choosing the sign of the square root if $H_{jk} \in \mathbb{R}^-$~\cite{BC12}.}.
Let $T \in \mathcal{U}(\mathscr{H}_2 \otimes \mathscr{H}_2^{\otimes n} \otimes \mathscr{H}_2^{\otimes (n+1)})$ 
be a unitary operator with action
\begin{equation}
T:\ket{0^{n+2}}\ket{j}\mapsto \ket{\psi_{j}} \quad \forall\, j \in [2^n],
\end{equation}
which prepares the state $\ket{\psi_{j}}$ conditional on the third register being in the state $\ket{j}$.
Let 
\begin{equation}
R_T := 2 T T^{\dagger}-\mathds{1}=\sum_{j=1}^{2^n}|j\rangle\langle j| \otimes\left(2\left|\varphi_{j}\right\rangle\left\langle\varphi_{j}\right|-1\right)
\end{equation}
be the controlled Householder reflection~\cite{Hou1958}, i.e.\ a reflection about $\left|\varphi_{j}\right\rangle$ conditional 
on the state $\ket{j}$ of the third register. 
Let $U_{\rm swap} \in \mathcal{U}(\mathscr{H}_2 \otimes \mathscr{H}_2^{\otimes n} \otimes \mathscr{H}_2^{\otimes n+1})$
be the operator that swaps the first $n+1$ qubits with the remaining $n+1$ qubits, respectively. Then a 
 $\left(d_H\beta_H, n+2, \epsilon\right)$ block encoding of $B$ is given by the unitary operator (Lemma 6, \cite{LC19})
\begin{equation}\label{eq:walkoperator}
    U_H = T^\dagger U_{\rm swap}R_T T.
\end{equation}
We next explain how to implement this unitary operator. 

A procedure for preparing $\left|\varphi_{j}\right\rangle$ from the $\ket{0}$ state 
is key to performing $T$ and $R_T$. 
A reflection about $\left|\varphi_{j}\right\rangle$ is achieved by performing inverse state preparation, 
reflecting about $|0\rangle,$ and then performing state preparation. Note that implementing $T$ is equivalent to preparing the 
state $\left|\varphi_{j}\right\rangle$ conditional on the third register being in the
state $\ket{j}$. 
Following the procedure given in the proof of Lemma~10~\cite{BCK15}, 
the operator $T$ can be implemented as follows.
First, a Hadamard gate is performed on 
$\lceil\log{d_H}\rceil$ qubits, and then the oracle $H_{\rm loc}$ is applied to obtain the state
\begin{equation}
    \frac{1}{\sqrt{d_H}} \sum_{\ell \in [d_H]}\ket{0}\ket{H_{\rm loc}(j,l)}.
\end{equation}
The final step consists of using the  oracle $O_{H_{\rm val}}$ and applying controlled rotation to the ancilla qubit
to prepare the desired state. The unitary operators 
$T$, $T^\dagger$ and $W$ can be implemented up to error at most $\epsilon$ in operator norm using $\mathcal O(1)$ 
queries to the oracles encoding $H$ and $\mathcal O\left(n +\log ^{2.5}\left(d_H\beta_H / \epsilon\right)\right)$ additional $2$-qubit gates~\cite{BC12}. As $U_{\rm swap}$ requires $\mathcal{O}(n)$ 2-qubit gates, the complexities follow readily.

\end{proof}

\noindent This lemma concludes our discussion of block access to a Hermitian matrix.

\subsubsection{Matrix arithmetic in the block-encoding setting}
We now review the matrix arithmetic techniques that we use for constructing a quantum algorithm for
solving B-SLEP and for constructing our lower bound. 
Quantum algorithms for solving a system of linear equations can be reduced to using a block encoding of 
a matrix~$A$ to construct a block encoding of $A^{-1}$. The block encoding of $A^{-1}$ can be obtained by an algorithm for quantum signal processing~\cite{LC19},
which maps a block encoding of a given matrix $H$ 
to that of $\mathcal P(H)$ for some given polynomial $\mathcal P$ of fixed degree and fixed coefficients.
We also review a result that provides a polynomial
for approximating the inversion function $\operatorname{inv}: x \mapsto 1/x$.

The complexity of quantum signal processing is given 
in the following theorem.
\begin{thm}[\cite{Gilyn2019}] \label{thm:P(H)} Given block access $(\alpha_A, a_A, \delta_A, U_A)$ to a 
Hermitian matrix $A$ of any size, an error $\sigma > 0$ and
a polynomial $\mathcal P: \mathbb R \to \mathbb R$ of degree $d$ satisfying $|\mathcal P(x)| \leq 1/2$ for all $x \in [-1,1]$, 
then block access $(1,a_A+2, 4d\sqrt{\delta_A/\alpha_A}  + \sigma, U_{\mathcal P(A/\alpha_A)})$ to  $\mathcal P(A/\alpha_A)$ 
can be constructed such that $U_{\mathcal P(A/\alpha_A)}$ makes
 $2d+1$ queries to $U_A$ and $\mathcal O(a_Ad)$
additional 2-qubit gates.
\end{thm}

\noindent The original proof of Theorem \ref{thm:P(H)} constructs a circuit for implementing $U_{P(A/\alpha_A)}$. Circuits for 
$U^\dagger_{P(A/\alpha_A)}$, $\con U_{P(A/\alpha_A)}$ and $\con U^\dagger_{P(A/\alpha_A)}$ can be constructed using the 
procedure explained in the preliminaries. Note that $\con \con U_A$
can be constructed using $\con U_A$ and two Toffoli gates.

We use the following corollary about approximation of the function inv$(x)$ with polynomials.
\begin{coro}[\cite{Gilyn2019}]\label{thm:poly approx of 1/x}
For any $\epsilon, \delta \in (0,1/2]$, a polynomial $\mathcal P:\mathbb R \to \mathbb R$ with odd degree $\mathcal O\left(\frac{1}{\delta}\log\left(\frac{1}{\epsilon}\right)\right)$ exists such that for all $x \in [-1,1] \setminus [-\delta, \delta]$ then  $\left|\mathcal P(x)\right|\leq1$
and
\begin{equation}
    |\mathcal P(x) - f(x)| < \epsilon,
\end{equation}
where $f(x) = \frac{3\delta}{4}\operatorname{inv}(x)$.
\end{coro}

\noindent This corollary concludes our discussion of block-encoding techniques for matrix arithmetic. 

\subsection{Complexity of solving a system of linear equations}
\label{sec:F-SLEP}
We now review the results that describe the hardness of solving restrictions of SLEP with respect to accuracy.
We also state a previously known query-complexity lower bound for computing the mean of a black-box function 
that we later use to derive a new lower bound for SLEP. 

\subsubsection{Previous results on the accuracy dependence of system of linear equations}
\label{sec:backgroundlowerbound}
It is often asserted that quantum algorithms for solving SLEP must have $\poly(1/\epsilon)$ accuracy dependence due to the cost of  estimating $\bra {\bm x} M \ket {\bm x}$~\cite{HHL09,Amb12,Gilyn2019,Chia2020,L2021}. However, this assertion
is proven to hold only under certain restrictions. One of the existing results shows that $\polylog(1/\epsilon)$ 
dependence of the time complexity is not achievable under complexity-theoretic assumptions. Another  previous hardness result shows that any quantum algorithm with $\polylog{N}$ dependence has runtime in $\mathcal O(1/\epsilon)$; however, as this result is derived by bounding queries to oracles for $A$, it is not applicable if $A$ is fixed. We now review these two results in detail. 

We adopt the standard big-O definition for single variables and the following definition for multiple variables~\cite{BB1996}.
\begin{definition}\label{def: multivar big O}
Let $f,g:\{\mathbb{R}^+\}^m \to \mathbb{R}^+$. Then $f(x_1,x_2,\dots,x_m) \in \mathcal O(g(x_1,x_2,\dots,x_m))$ if there exist constants $\eta, C \in \mathbb R^+$ such that for all $x_1,x_2,\dots,x_m >\eta$, 
\begin{equation}
    f(x_1,x_2,\dots,x_m) \leq Cg(x_1,x_2,\dots,x_m).
\end{equation}
\end{definition}
\noindent We define multivariate $\Omega$ analogously.

We now rigorously define F-SLEP, a previously studied problem~\cite{HHL09} which we first reviewed in \S\ref{sec:intro}.
\begin{problem}[Fixed $M$ and $b$ System of Linear Equations Problem (F-SLEP)]
\label{prob:FSLEP}
Given $N\in\mathbb{Z}^+$, a $\kappa \geq 1$, sparse access $(d_A, \beta_A, 
O_{A_{\rm val}},O_{A_{\rm loc}})$ to an $N \times N$ matrix $A$ with  
$\|A\| \leq 1$ and $\left\|A^{-1}\right\| \leq \kappa$, and an accuracy $\epsilon\in (0,1]$, 
return an $\epsilon$-additive approximation of $\langle \bm{x}|M| \bm{x}\rangle$ with probability at least $2 / 3$, 
where $\ket{\bm{x}}$ is the unit vector proportional to $A^{-1}\ket{\bm{0}}$ and $M$ is the diagonal matrix containing a 1 in the first $N/2$ diagonal entries and a zero in all other entries.
\end{problem} 

Two hardness results are known for F-SLEP~\cite{HHL09}. 
The first result establishes a lower bound on the oracle queries to $A$, and its proof relies on
a query-complexity lower bound for computing the parity of a Boolean expression.
\begin{thm}[\cite{HHL09}]
\label{lem:HHLlowerbound}
A quantum algorithm can solve F-SLEP using $\nu$ queries to oracles for $A$ with
$\nu \in N^{c_1}\poly(\kappa)/ \epsilon^{c_2}$ only if $c_1+c_2 \geq 1$.
\end{thm}
\noindent The original statement of this theorem refers to the time complexity of quantum algorithms for SLEP rather than number of queries \cite{HHL09}. However, the proof given is for a stronger statement in terms of queries to oracles for $A$ that we give in Theorem \ref{lem:HHLlowerbound}. 

The next hardness result for F-SLEP is in terms of time complexity and 
relies on the complexity theoretic assumption $\textbf{BQP} \neq \textbf{PP}$.
\begin{thm}[\cite{HHL09}]
\label{lem:HHLlowerbound2}
If $\textbf{BQP} \neq \textbf{PP}$, then there does not exist a quantum algorithm that solves F-SLEP running in time  
$\poly(\kappa, \log(N), \log(1/\epsilon))$.
\end{thm}
\noindent 
In fact, the proof of this theorem proves the complexity $\polylog(N,1/\epsilon)$ for the restriction $A=\mathds{1}$, 
which is a stronger result.
Note that $\textbf{BQP} = \textbf{PP}$ contradicts the widely held conjecture $\textbf{BQP} \ne \textbf{NP}$. 

\subsubsection{Query complexity of approximating the mean of a black-box function}

We now review a query-complexity lower bound for computing the approximate mean of a black-box function~\cite{NW99}.
We use this lower bound to later derive a lower bound for B-SLEP in \S\ref{sec:lowerbound}. As is standard, we say an oracle $O_f$ encodes~\cite{NC} a discrete function  $f:\{0,1\}^{r_1}\to \{0,1\}^{r_2}$ for $r_1,r_2 \in \mathbb{Z}^+$ if
\begin{equation}
    O_f: \ket{j}\ket{s} \mapsto \ket{j}\ket{f(j) \oplus s}, \;\; \forall j \in \{0,1\}^{r_1}, s \in \{0,1\}^{r_2},
\end{equation}
 where $\oplus$ denotes bitwise XOR operation.
 
 We are now ready to state the approximate-mean (AM) problem. 
\begin{problem}[Approximate Mean] 
\label{prob:AM}
Given $N \in \mathbb{Z}^+$ and $\epsilon \in (1/2 N,1]$ and given an oracle $O_f$ encoding a function 
$f:[N] \mapsto [0,1]$, return, with probability at least 2/3, an $\epsilon$-additive approximation of $\mu_f := \left(\sum_{j=0}^{N-1} f(j)\right)/N$.
\end{problem}

A query-complexity lower bound for this problem is given by Nayak and Wu~\cite[Corollary 1.12]{NW99}.
\begin{coro} \label{lem:queryAM}
The query complexity of the problem AM is $\Omega(1/\epsilon)$.
\end{coro}
\noindent 
This concludes our review of relevant hardness results.

\subsection{A reformulation of amplitude-estimation algorithm}
\label{sec:backgroundae}
We now review a reformulation~\cite{KOS07} of the amplitude-estimation algorithm~\cite{BHM+02}.
This reformulation is used in some quantum algorithms for expectation-value estimation that we review in this subsection. 
We begin by reviewing the complexity of the standard phase-estimation algorithm~\cite{NC}. 
Next, we describe how phase estimation is used to construct the amplitude-estimation algorithm. 

Phase estimation is a key subroutine in many quantum algorithms that achieve super-polynomial speedup over classical algorithms. The inputs, output, and complexity of the phase-estimation algorithm are given by the following theorem, which is an adaptation from the standard description of the algorithm \cite{NC}. 
 
\begin{thm}[Phase-estimation Algorithm]
\label{thm:PEA}
There exists a quantum algorithm that accepts  $n$-qubits, an accuracy $\epsilon$, and $n$-qubit unitary black boxes $V$ and $W$, that, with probability at least 2/3, samples an $\epsilon$-additive approximation to the eigenphase $\theta$ of $V$ corresponding to the eigenvector $\ket{\psi}$ with probability $\abs{\braket{\psi| W|0}}^2$. This algorithm makes $
  \mathcal  O\left(1/\epsilon\right)$
queries to $V$, a single query to $W$, 
and  $
  \mathcal O\left(\log^2(1/\epsilon)\right)$
additional 2-qubit gates. 
\end{thm}

Now we show how phase estimation is used to perform amplitude estimation. In amplitude estimation~\cite{KOS07}, the inputs are $n$ qubits, 
an accuracy $\epsilon \in (0,1)$,  
and two $n$-qubit unitary black boxes  $V$ and $W$. The output is an $\epsilon$-additive approximation to the amplitude $r:=\abs{\braket{0|W^\dagger VW|0}}$ 
with probability at least $2/3$.
Define $\ket{\psi_0} := W\ket 0,\;  
 \ket{\psi_1}:=V\ket{\psi},\; P_0:= \mathds1 -2\ket{0^n}\bra{0^n}$, and
define two reflection operators $S_0:= \mathds1-\ket{\psi_0}\bra{\psi_0}=WP_0W^\dagger$ and $S_1:=\mathds1-\ket{\psi_1}\bra{\psi_1}= VWP_0W^\dagger V^\dagger$.
The composition of these two reflection operators $S:=S_0S_1$ is a rotation operator in the two-dimensional space spanned by $\ket{\psi_0}$ and $\ket{\psi_1}$ that rotates $\ket{\psi_0}$ toward $\ket{\psi_1}$ by angle $2\theta:=4\arccos(r)$.
This implies that the eigenvalues of $S$ are $\text{e}^{\text{i}\theta}$, and the amplitude returned by amplitude estimation is $r=\abs{\cos(\theta/2)}$. Therefore, phase estimation can be used to estimate $r$ by estimating an eigenphase of the operator $S$. Note that this call to phase estimation uses $\con S$, which can be constructed by replacing each factor of $V$ and $W$ in the definition of $S$ with $\con V$ and $\con W$.

\begin{algorithm}[H]
\begin{algorithmic}[1]
\caption{Amplitude-Estimation Algorithm } \label{alg:AEA}
\Require{
A number of qubits $n\in\mathbb{Z}^+$, an accuracy $\epsilon \in (0,1)$, 
 two $n$-qubit unitary black boxes $V,W$
}
\Ensure{
With probability at least 2/3, an $\epsilon$-additive approximation $\tilde r$
to the amplitude $r = \abs{\braket{0 |W^\dagger V W| 0}}$
}

\hspace{-1.3cm}\ \textbf{Procedure:}

   \State Construct a quantum circuit for $\con S$ where $S =  WP_0W^\dagger VWP_0W^\dagger V^\dagger$ and $P_0 = \mathds 1 -2\ket{0^n}\bra{0^n}$.
    \State Apply phase estimation with input $(n,2\epsilon,\con S,W)$ 
    to obtain with probability at least $2/3$ 
    a $(2\epsilon)$-additive approximation $\tilde \theta$ of some eigenphase $\theta$ of $S$. 
    \State \Return $\tilde r=\abs{\cos( \tilde \theta/2)}$.
\end{algorithmic}
\end{algorithm}
The complexity of this algorithm is given in the following lemma, which follows readily from the analysis in Ref.~\cite{KOS07}. 
\begin{lemma}
\label{lem:AEcomplexity}
Algorithm~\ref{alg:AEA} with input ($\epsilon,V,W$) makes $
    \mathcal O\left(\frac 1 \epsilon\right)
$
queries to $V$ and $W$ and uses
$ \mathcal O\left(\frac{n}{\epsilon}\right)$
additional 2-qubit gates.
\end{lemma}
\begin{proof}
The complexity for queries to $V$ and $W$ follows directly from the complexities in Theorem \ref{thm:PEA}.
Note that $P_0$ can be implemented using $\mathcal O(n)$ 2-qubit gates \cite[p.\ 251]{NC} which means the call to the phase-estimation algorithm uses a total of $\mathcal O(n/\epsilon)$ 2-qubit gates for queries to $P_0$ and $\log^2(1/\epsilon)$ 2-qubit gates for other operations. 
\end{proof}
\noindent This lemma concludes our review of amplitude estimation.

\subsection{Quantum algorithms for computation of expectation value}
Quantum algorithms for estimating expectation values have been developed both in the context
of SLEP~\cite{WBL12, HHL09, Wan17}  and in a variety of other applications~\cite{KOS07, HWM0W2021, JW2007, Ral2020, KL2021,HWM0W2021}. 
We now review quantum-algorithmic techniques for
estimating expectation values and discuss limitations of these techniques.

\subsubsection{Estimating expectation values in restrictions of SLEP}
We begin by discussing techniques for estimating expectation values that occur in quantum algorithms for SLEP. We review three
techniques. These techniques use single-qubit computational-basis measurements, the swap test, and phase estimation, respectively. 

 We now discuss the first technique. In F-SLEP (Problem~\ref{prob:FSLEP}), $M$ is the matrix representation of the operator  $ \ket{0}\bra{0}\otimes\mathds{1}_{{n-1}}$ which is the orthogonal
projector on the first qubit. As a result, the expectation value of a state $\ket{\bm{x}}$ with respect to $M$ can be estimated by repeatedly preparing $\ket{\bm{x}}$, measuring the first qubit in the computational basis, and averaging the
outcomes~\cite{HHL09}. This approach requires $\mathcal O(1/\epsilon^2)$ preparations of $\ket{\bm x}$ to obtain an estimate
of the expectation value up to additive error $\epsilon$~\cite{HHL09}. 

We now discuss the second technique. If $M = A^{-1}$, where $A$ defines the system of linear equations (\ref{eq:Ax = b}), then a swap test can be used to estimate $\braket{\bm{x}|M|\bm{x}} = \braket{\bm{x}|\bm{b}}$, as was demonstrated in an application of the HHL algorithm to least-square fitting~\cite{WBL12}. This estimation is achieved by 
first preparing the state $\ket{+}\ket{\bm{b}}\ket{\bm{x}}$ using an algorithm for QLSP, with $\ket{+}:=(\ket{0}+\ket{1})/\sqrt{2}$,
followed by a controlled-swap operation, with control on the first qubit, to obtain the state 
$(\ket{0}\ket{\bm{b}}\ket{\bm{x}} + \ket{1}\ket{\bm{x}}\ket{\bm{b}})/\sqrt{2}$. Finally, the first (control) qubit
is measured in the computational basis. The probability that this measurement returns the outcome $0$
is $(1 + \abs{\braket{\bm{x}|\bm{b}}}^2)/2$. To obtain an estimate of $\abs{\braket{\bm{x}|\bm{b}}}^2$
to additive error~$\epsilon$, the swap test needs to be carried out $\mathcal O(1/\epsilon^2)$ times.

If every entry of $M$ is zero except for a single 1 on its $i$th diagonal entry, then the expectation value
$\braket{\bm{x}|M|\bm{x}} = \abs{\braket{\bm{x}|i}}^2$ can be estimated by a direct application of 
the amplitude-estimation algorithm~\cite{BHM+02}. This technique was used for constructing an algorithm
for computing the effective resistance in a given electrical network~\cite{Wan17,Gilyn2019}.
The time complexity of this method scales as $\mathcal O(1/\epsilon)$ with respect to
the additive-error $\epsilon$. 

The key limitation of all three techniques discussed above is that they each apply only to some limited class of matrices; the first technique  applies if $M$ represents the orthogonal projector on the first qubit, the second technique applies  if $M = A^{-1}$, and the final technique applies if $M$ has entries that are all zero except for a single 1 on the diagonal. Moreover,
the first two techniques, which rely on classical sampling, achieve $\mathcal O(1/\epsilon^2)$  rather than
$\mathcal O(1/\epsilon)$.

\subsubsection{Estimating expectation values of simulatable Hermitian matrices}

We now discuss techniques for estimating expectation values, all of which are outside of the context of solving SLEP. 
The four techniques rely on the Hadamard test and algorithms for Hamiltonian simulation, phase estimation and amplitude estimation.

We now review the first technique, which is a quantum algorithm for estimating $\braket{0^n|V^\dagger UV |0^n}$, where $U$ and $V$ 
are given $n$-qubit unitary black boxes. 
One  common method for this estimation is the Hadamard test and proceeds as follows~\cite{AJL2008}.  First prepare the state 
$\ket{+}\ket{\psi}$, and then apply the controlled unitary operator $\con U$ (with control on the
first qubit) to obtain the state $(\ket{0}\ket{\psi}+\ket{1}U\ket{\psi})/\sqrt{2}$, where $\ket \psi := V\ket 0$. Next, 
apply a Hadamard gate $H$ on the first qubit and measure this qubit in the computational basis. 
The expectation value of the output is $\text{Re}\braket{\psi|U|\psi} = \text{Re}\braket{0^n|V^\dagger UV |0^n}$. To obtain an $\epsilon$-additive 
approximation of the real part of the expectation value, $\mathcal O(1/\epsilon^2)$ runs of the Hadamard test are required. 
The imaginary part of the expectation value can be estimated similarly by instead starting with the state 
$(\ket{0}-\text{i}\ket{1})\ket{\psi}/\sqrt{2}$.
    
The second technique for estimating the expectation value of $U$ relies on  amplitude estimation~\cite{KOS07}. Reduction to amplitude estimation makes use of the relations
\begin{eqnarray}
\label{ev2ae}
\text{Re}(\braket{\psi|U|\psi}) &  = & (4\abs{\bra{+}\bra{\psi}\con U \ket{+}\ket{\psi}}^2 - \abs{\braket{\psi|U|\psi}}^2 
- 1)/2, \nonumber \\
\text{Im}(\braket{\psi|U|\psi}) &  = & \left(4\abs{\bra{+}\bra{\psi}(\text{e}^{\text{i}\sigma_z\pi/4}\otimes \mathds{1}_n)\con U \ket{+}\ket{\psi}}^2 
- \abs{\braket{\psi|U|\psi}}^2 - 1\right)/2.
\end{eqnarray}
The terms $\abs{\bra{+}\bra{\psi}\con U \ket{+}\ket{\psi}}^2$, $\abs{\bra{+}\bra{\psi}(\text{e}^{\text{i}\sigma_z\pi/4}\otimes 
\mathds{1}_n)\con U \ket{+}\ket{\psi}}^2$ and $\abs{\braket{\psi|U|\psi}}^2$ 
on the right-hand sides of Eq.\,\eqref{ev2ae} can each be estimated by making one call each to the amplitude-estimation algorithm. 
We review  amplitude estimation  in \S\ref{sec:backgroundae}. This approach yields an algorithm 
for computing the expectation value $\braket{\psi|U|\psi}$ with complexity $\mathcal O(1/\epsilon)$.

Whereas the two techniques above apply only to estimating the expectation value of unitary matrices,
they are used as subroutines in other techniques for estimating expectation values of Hermitian matrices.
We now review a technique for estimating the expectation value of any simulatable matrix $M$~\cite{KOS07}.
A $N \times N$ matrix $M$ is simulatable if, for any $t \in \mathbb{R}^+$ and $\epsilon \in (0,1)$, 
the propagator $\text{e}^{\text{i}Mt}$ can be implemented by a $\poly(\log(N),t,1/\epsilon)$-sized circuit~\cite{ATS2003}. 
The expectation value $\braket{\psi|M|\psi}$ is then obtained by using the approximation~\cite{KOS07}
\begin{equation}
    \abs{\braket{\psi|\text{e}^{-\text{i}Mt}|\psi} - 1 + \text{i}t\braket{\psi|M|\psi}} \in  \mathcal O(t^2).
\end{equation}
The term $\braket{\psi|\text{e}^{-\text{i}Mt}|\psi}$ is computed by using one of the techniques for 
estimating the expectation value of a unitary matrix, which we discussed above.
If $M$ is simulatable by a $\poly(n,t,\log(1/\epsilon))$-sized circuit, then $\braket{\psi|M|\psi}$
can be computed to additive approximation $\epsilon$ in $\soft {\mathcal O}((1/\epsilon)^{1+\alpha})$ time for arbitrarily small $\alpha$
using an improved version of this technique. 
Here $\soft{\mathcal O}$ is called soft O and indicates suppression of logarithmic factors. 
    
An alternative technique for estimating the expectation value of a Hermitian and simulatable $M$
employs the phase-estimation algorithm applied to the operator $\text{e}^{\text{i}Mt}$~\cite{JW2007}. Suppose $M$ has spectral decomposition
$M = \sum_{j}\lambda_j \ket{v_j}\bra{v_j}$, and suppose $\ket \psi$ has decomposition $\ket{\psi} = \sum_j \alpha_j\ket{v_j}$.
One application of the phase-estimation circuit results in a state close to
$\ket{\psi_M} =  \sum_j \alpha_j\ket{v_j}\ket{\lambda_j}$. Measuring the final register in the computational basis
yields $\lambda_j$ with probability approximately $\abs{\alpha_j}^2$ for each $j$. The expectation value
$\braket{\psi|M|\psi} = \sum_j \abs{\alpha_j}^2\lambda_j$ can then be estimated by repeating the above procedure
$\mathcal O(1/\epsilon^2)$ times and averaging the measurement outcomes. Assuming that $M$ is simulatable by a 
$\poly(n,t,\log(1/\epsilon))$-sized circuit, the total complexity of this algorithm is evidently $\soft {\mathcal O}(1/\epsilon^3)$, which is not explicitly stated.

In conclusion, the first two techniques are particularly relevant to our work only for building the latter two techniques. These final two techniques for estimating the expectation value  outside of the context of SLEP 
apply to simulatable Hermitian matrices, which include all sparse matrices~\cite{ATS2003}. The best of these techniques 
achieves close to linear dependence on~$1/\epsilon$.
These techniques have not been combined with the algorithms for QLSP to solve SLEP.

\subsubsection{Estimating expectation values of block-encoded Hermitian matrices}
\label{sec:backgroundbevhm}
In this section, we review an algorithm for estimating the expectation value of a block-encoded Hermitian matrix
with respect to a quantum state given by a unitary black box~\cite{Ral2020}. The algorithms developed for
estimating the expectation value of simulatable Hermitian matrices can also be used for estimating the expectation 
values of block-encoded matrices. However, there is an elegant and more efficient way to estimate the expectation value of
block encoded matrices.

We first adapt the following problem
from Lemma 5 in Ref.~\cite{Ral2020}.
\begin{problem}[Block-access Expectation Value of a Hermitian Matrix (B-EVHM)] \label{pro:B-EVHM}
Given~$n$ qubits,
block access $(\alpha_M, a_M,0, U_M)$ to a $2^n\times 2^n$ Hermitian matrix $M$,  an accuracy $\epsilon \in [\alpha_M/2^n , \alpha_M]$,
and an $n$-qubit unitary black-box~$V$, 
return with probability at least $2/3$ an $\epsilon$-additive approximation to $\braket{{0^n}|V^\dagger M V|{0^n}}$. 
\end{problem}
\noindent We now review an algorithm for B-EVHM given in Ref.~\cite{Ral2020} that achieves linear dependence on $1/\epsilon$. 
We include a full analysis of the complexity of this algorithm by counting queries to 
$V$ and $U_M$, and additional two-qubit gates separately.
To simplify this analysis, we employ the standard amplitude-estimation algorithm, 
as opposed to the more recent amplitude-estimation algorithms which do not require the quantum Fourier transform \cite{GGZW19}. 
This change in amplitude-estimation algorithm has no effect on complexity. 

As the expectation value of $M$ is not always positive,
it cannot be computed by a direct application of amplitude estimation. The problem is first reduced 
to one solvable by amplitude estimation
in the following lemma, using a technique that employs a controlled block encoding~\cite{KOS07}. 
\begin{lemma}
\label{lem:EVtoAE}
For any $n$-qubit Hermitian matrix $M$ with $\norm{M}\le \Lambda$ and for any $n$-qubit state $\ket{\psi}$,
\begin{equation}
\label{eq:McM}
    \braket{\psi|M|\psi} =
    \Lambda\left(2\abs{\bra{+}\bra{\psi} \Con  (M/\Lambda) \ket{+} \ket{\psi}}-1\right).
\end{equation}
\end{lemma}

\begin{proof}
The matrix $M$ is Hermitian; therefore, $\braket{\psi|M|\psi}\in\mathbb{R}$ for any $\ket{\psi} \in \mathscr{H}_2^{\otimes n}$.
As $\norm{M/\Lambda} \le 1$, we have $\abs{\braket{\psi|M/\Lambda|\psi}} \le 1$ and
$\braket{\psi|M/\Lambda|\psi}+1 = \abs{1+\braket{\psi|M/\Lambda|\psi}}$.
Eq.~(\ref{eq:McM}) follows from
\begin{align}
    2|\bra{+} \bra{\psi} \Con (M/\Lambda) \ket{+} \ket{\psi}|\nonumber 
    &= 2|\bra{+}\bra{\psi}\left(\ket{0}\bra{0}\otimes \mathds{1} + 
    \ket{1}\bra{1}\otimes M/\Lambda \right) \ket{+}\ket{\psi}| \nonumber \\
    & = |1+\braket{\psi|M/\Lambda|\psi}| \nonumber \\
    & = {1+\braket{\psi|M/\Lambda|\psi}}.
\end{align}
\end{proof}

From the definition of block encoding (Def. \ref{def:block encoding}) it is evident that $\|M\| \leq \alpha_M$ for any block encoding. For our application of Lemma \ref{lem:EVtoAE}, $\alpha_M$, which is included in block access to $M$, is used as the bound $\Lambda$ on the norm of $M$.
The next lemma, which is a special case of Lemma 52 in Ref.~\cite{Gilyn2019}, shows that block access to $\Con (M/\alpha_M)$ can be implemented if block access
to $M$ is given.
\begin{lemma}[\cite{Gilyn2019}]
\label{lem:controlM}
Given $n$ qubits and block access $(\alpha_M,a_M,0, U_M)$ to an $2^n \times 2^n$ Hermitian matrix $M$. Then block access $\left(1, a_{M}, 0, U_{M'}\right)$ to the Hermitian matrix  $M' : = \Con (M/\alpha_M)$
can be implemented such that $U_{M'}$ uses a single query to $U_M$.
\end{lemma}
\begin{proof}
The operator $U_M$ is an $(n+a_M)$-qubit unitary. The desired circuit for  $U_{M'}$ acts on $a_M + 1 + n$ qubits and comprises only a single query to $\con U_M$ in which the $(a_M +1)$th qubit is the control qubit, and the $U_M$ operation controlled by this qubit operates on the first $a_M$ qubits and the last $n$ qubits.
We now show that the operator $U_{M'}$ implemented by this circuit is indeed a $(1, a_{M}, 0)$ block encoding of $\Con (M/\alpha_M)$. Let $s_1$ and $s_2$ be arbitrary $n$-bit strings. Then
\begin{equation}\label{eq: block check 1 c-M}
    \bra{0^{a_M},0,s_1}U_{M'}  \ket{0^{a_M},0,s_2} = \delta_{s_1, s_2},
\end{equation}
as the control qubit is set to $\ket{0}$, so $U_M$ is not applied.
Furthermore, 
\begin{equation}\label{eq: block check 2 c-M}
    \bra{0^{a_M},1,s_1}U_{M'}  \ket{0^{a_M},1,s_2} = M_{s_1,s_2}/\alpha_M,
\end{equation}
by definition of block encoding.
Finally, $\bra{0^{a_M},0,s_1}U_{M'}  \ket{0^{a_M},1,s_2} = 0$.
By equations (\ref{eq: block check 1 c-M}) and (\ref{eq: block check 2 c-M}),
\begin{equation}\label{eq: matrix of C-M}
     \bra{0^{a_M}}U_{M'} \ket{0^{a_M}} = \begin{pmatrix} 
     \mathds 1 & \mymathbb 0 \\
    \mymathbb 0 & M/\alpha
     \end{pmatrix},
\end{equation}
where $\mymathbb 0$ is a $2^n \times 2^n$ zero matrix.
As Eq.~(\ref{eq: matrix of C-M}) gives the matrix representation of $\Con (M/\alpha)$, the operator $U_{M'}$ is the desired block encoding.
\end{proof}

The following theorem states the complexity of solving B-EVHM. 
\begin{thm}
\label{thm:B-EVHM}
Given $n$ qubits, 
block access $(\alpha_M, a_M,0, U_M)$ to a $2^n \times 2^n$ Hermitian matrix $M$, an accuracy~$\epsilon > 0$ and a unitary black box $V$, then
an $\epsilon$-additive approximation $\tilde{u}$ to $u := \braket{0 |V^\dagger M V| 0 }$ 
can be returned, with probability at least 2/3, 
using $\mathcal  O(\alpha_M/\epsilon)$
queries to $U_M$, 
$\mathcal  O(\alpha_M/\epsilon)$
queries to $V$
as well as 
$\mathcal O\left((n+a_M)\alpha_M/\epsilon\right)$
additional 2-qubit gates.
\end{thm}
\begin{proof}
We prove this theorem by constructing Algorithm \ref{alg:AB-EVHM}, and now proceed to prove its correctness and complexity. 
\begin{algorithm}[H]
\begin{algorithmic}[1]
\caption{Algorithm for Block-access Expectation Value of an Hermitian Matrix \label{alg:AB-EVHM}}
\Require{
$n$ qubits, 
 block access $(\alpha_M, a_M,0, U_M)$ to $M$, an accuracy $\epsilon > 0$, and a unitary black box $V$}
\Ensure{
an $\epsilon$-additive approximation $\tilde{u}$ to $u := \braket{0|V^\dagger M V| \psi }$ 
with probability 2/3
}

\hspace{-1.3cm}\ \textbf{Procedure:}
\State Construct block access $(1, a_M,0, U_{M'})$ to  $M' = \Con (M/\alpha_M)$ 
\Comment{use Lemma \ref{lem:controlM}}.
\State Compute $\tilde{r}$ by executing Algorithm \ref{alg:AEA} with input $(n+a_M,\epsilon/2\alpha_M,U_{M'},
\mathds{1}_{a_M}\otimes H \otimes V)$.
\State \Return $\tilde{u} := \alpha_M(2\tilde{r}-1)$.
\end{algorithmic}
\end{algorithm}

\noindent \textit{Correctness}:
Define $r := \abs{\bra{+}\bra 0 V^\dagger M' V\ket{+}\ket{0}}$.
Then $\tilde{r}$ is an $(\epsilon/2\alpha_M)$-additive approximation to $r$ with probability at least 2/3.
 By Lemma~\ref{lem:EVtoAE}, using $\alpha_M$ as an upper bound on $\|M\|$,  we have  $u = \alpha_M(2r-1)$.
Therefore, $\tilde{u}$ is an $\epsilon$-additive approximation to $u$ with probability at least 2/3.

\noindent \textit{Complexity:}
As the algorithm makes only a single call to amplitude-estimation, the complexities follow directly from Lemma~\ref{lem:AEcomplexity} and Lemma~\ref{lem:EVtoAE}.
\end{proof}

Algorithm \ref{alg:AB-EVHM} concludes our review of algorithms for estimating expectation values of block-encoded matrices.

\section{Approach}
\label{sec:approach}

In this section, we describe our formulation of SLEP in the block encoding setting,
our methods for establishing query-complexity lower bounds for this computational problem, 
and the techniques we use in constructing a quantum algorithm that saturates this bound.
The detailed derivation of our results is presented in~\S\ref{sec:results}.

\subsection{Query-complexity lower bound for B-EVHM}
\label{sec:approachlowerbound}

We begin by presenting our computational problem statements, starting with B-SLEP, which is a formulation of SLEP in the block-encoding setting. We then discuss two expectation-value-estimation problems. We close by
describing the reductions we use to derive our lower bounds for B-SLEP.

In contrast to QLSP~(Problem~\ref{prob:SQLSP}), the output of SLEP is a number, not a quantum state. We formulate SLEP in the block-encoding setting.
\begin{problem}[Block-access System of Linear Equations Problem (B-SLEP)]
\label{pro:SLE1}
Given $n$ qubits, 
a $\kappa \ge 1$,
block access $(\alpha_A,a_A,0, U_A)$ to a $2^n\times 2^n$ invertible Hermitian matrix $A$ such that $\norm{A^{-1}}\alpha_A \le  \kappa$, 
block access $(\alpha_M,a_M,0, U_M)$ to a $2^n\times 2^n$ Hermitian matrix $M$, 
an accuracy $\epsilon \in [\alpha_M/2^n,\alpha_M]$,
and an $n$-qubit unitary black box $U_{\bm b}$, 
return with probability at least $2/3$ an $\epsilon$-additive approximation to $\bm x^\dagger M \bm x$, where 
$\bm{x} := A^{-1}{\bm{b}}$ and $\bm{b}$ is the $2^n$-dimensional complex vector with entries $b_i = \braket{i|U_{\bm{b}}|0^n}$.
\end{problem}

We require that $\epsilon \le \alpha_M$ as otherwise the problem is trivial
with $0$ a valid answer. We restrict to instances with $\epsilon \ge \alpha_M/2^n$, as our query-complexity lower bound only applies to these instances. Additionally, by Theorem~\ref{lem:HHLlowerbound2}, any quantum algorithm solving B-SLEP for instances with $\epsilon \le \alpha_M/2^n$ must have complexity super-polynomial in~$n$ across these instances, under complexity-theoretic assumptions. As a result, the requirement $\epsilon \ge \alpha_M/2^n$ restricts only to instances for which a quantum algorithm could have complexity in $\poly(n)$. Note that $\alpha_A$ and $\alpha_M$ bound the norm of $A$ and $M$, as blocks of $U_A$ and $U_M$ do not have norm greater than $1$. 

It follows from the Chernoff bound \cite[p.\,154]{NC} that a constant number of repetitions of 
an algorithm solving Problem \ref{pro:SLE1} is 
sufficient to boost the success probability to a constant $s \in [2/3,1)$ using a number of repetitions that is logarithmic in $1/(1-s)$. Therefore, if the bound $2/3$ on the success probability is replaced by any constant in the interval $[2/3,1)$ no changes in complexity occur.
In any problem where the output is an $\epsilon$-additive approximation, 
we require that the output of the algorithm in the failed attempts is a complex number but not necessarily an $\epsilon$-additive
approximation. Although $A$ is assumed to be Hermitian in B-SLEP, our results also apply to non-Hermitian
matrices, as they can be encoded into Hermitian matrices~\cite{HHL09}.

We aim to establish a lower bound on the number of queries to $U_M$ that any quantum algorithm for B-SLEP must make.
As the solution state $\ket{\bm x}$ can be generated using algorithms for QLSP without making any queries to $U_M$, we focus on the measurement step where queries to $U_M$ are required.
The measurement step is in fact formalized by the problem B-EVHM discussed in \S\ref{sec:backgroundbevhm}.  
This problem is equivalent to B-SLEP restricted to the case $A = \mathds{1}$. 
Therefore, a lower bound for B-EVHM gives a lower bound for B-SLEP. 

Our lower bound for B-EVHM is found by first proving a lower bound on the analog of B-EVHM in the sparse-access setting, which
we state below.
\begin{problem}[Sparse-access Expectation Value of a Hermitian Matrix (S-EVHM)] \label{pro:S-EVHM}
Given  number of qubits $n$, block access $(d_M, \beta_M, O_{M_{\rm val}},O_{M_{\rm loc}})$ to a $2^n \times 2^n$ Hermitian matrix $M$, an accuracy $\epsilon \in [d_M\beta_M/2^n , d_M\beta_M]$, 
and an $n$-qubit unitary black box~$V$, 
return with probability at least $2/3$ an $\epsilon$-additive approximation to $\braket{{0^n}|V^\dagger M V|{0^n}}$. 
\end{problem}
To establish our query-complexity lower bound on S-EVHM, we reduce the problem of calculating the approximate mean 
of a black-box function (Problem~\ref{prob:AM}) to S-EVHM.
This reduction allows us to use the known bounds for AM~\cite{NW99}, given in Corollary \ref{lem:queryAM}, to derive bounds for S-EVHM. Furthermore, as S-EVHM is reducible to B-EVHM, and since B-SLEP is reducible to B-EVHM, 
we can extend these lower bounds to both B-EVHM and B-SLEP. 
The reduction of AM to S-EVHM relies on the observation that the 
expectation value of a Hermitian matrix $M$ with respect to the
uniform superposition state 
$\ket{+^n} := {1}/{\sqrt{N}}\sum_{j\in [N]}\ket{j}$ equals the mean of the entries of $M$, i.e.
\begin{equation}
\bra{+^n}  M  \ket{+^n} = \frac{1}{N}\sum_{i,j \in [N]}M_{ij}.
\end{equation}

Given oracle access to a function $f:[N]\to\{0,1\}$, we construct sparse access 
to a matrix $M$ that encodes the function $f$ in its entries. 
The mean of $f$ can then be calculated by one call to S-EVHM with $U = H^{\otimes n}$, which completes the reduction. 
A function $f$ can be encoded in the entries of a Hermitian matrix in many different ways, with one way being 
encoding the function values along the diagonal entries of the matrix. Whereas such an encoding would 
be sufficient to derive a lower bound for S-EVHM with respect to $\epsilon$ and the max-norm $\beta_M$, we employ a different reduction
to achieve a lower bound with respect to $\epsilon$, $d_M$ and $\beta_M$. We encode a function 
on domain $[d_M2^{n-1}]$ in the entries of a Hermitian matrix of sparsity $d_M$. 
We use this encoding to prove that the query complexity of S-EVHM is $\Omega(d_M \beta_M/\epsilon)$. 
Our query-complexity lower bound for B-EVHM is derived by a reduction
leveraging known methods for constructing  block encodings from  sparse-access 
encodings given in Lemma \ref{lem:block encoding sparse access}.
For B-EVHM, S-EVHM and B-SLEP, the query-complexity lower bound yields $\Omega(1/\epsilon)$ dependence on precision,
which rules out the possibility of any algorithm solving any of these problems in super-logarithmic time.

{ We note that a lower bound on B-EVHM could be established directly from the lower bound for
AM. However, we opt to first prove a lower bound on S-EVHM and then derive our lower bound on B-EVHM. 
Our motivation for this choice is that the lower bound for S-EVHM is important beyond its application to deriving 
a lower bound for B-SLEP, as estimation of expectation value of a sparse Hermitian matrix 
is a standard subroutine in several algorithms~\cite{KOS07, HWM0W2021, JW2007, Ral2020, KL2021,HWM0W2021}.}

\subsection{Tightness of the lower bound}
\label{sec:approachreduction}
Our approach for proving the tightness of our query-complexity lower bounds for B-EVHM, B-SLEP and S-EVHM 
is to give quantum algorithms that achieve these bounds. The algorithm for B-EVHM discussed in \S\ref{sec:backgroundbevhm}
saturates our bound for B-EVHM. We now describe 
how we design our quantum algorithms for B-SLEP and S-EVHM. 
The detailed construction of these algorithms, along with their complexities, 
is given in~\S\ref{subsubsec:B-SLEP} and \S\ref{app:sslep}.

We start by describing a reduction of B-SLEP to B-EVHM, and then describe our quantum algorithm for B-EVHM. 
This reduction employs a block-encoding of $A^{-1}$, which 
is known to be implementable with $\polylog(1/\epsilon)$ queries to $U_A$ using
techniques reviewed in \S\ref{sec:block encodings}.
Let $U_{A^{-1}}$ be an $(\alpha,a,0)$ block encoding of $A^{-1}$.
Observe that
\begin{equation}
\label{reduction}
    \bra{0^{a}}\bra{\bm b} U_{A^{-1}}^\dagger  \left( \ket{0^{a}}\bra{0^{a}} \otimes M \right)U_{A^{-1}}\ket{0^{a}}\ket{\bm b} = \frac{1}{\alpha^2}
     \bra{\bm{b}} {A^{-1}}^\dagger M A^{-1}\ket{\bm{b}} = \frac{1}{\alpha^2}\bm{x}^\dagger M\bm{x}.
\end{equation}
A block encoding of $\ket{0^{a}}\bra{0^{a}} \otimes M$ is implementable by a simple circuit that queries
the block encoding of $M$.
The left-hand side of Eq.~\eqref{reduction} is the expectation value of $\ket{0^{a}}\bra{0^{a}} \otimes M$
with respect to $U_{A^{-1}}\ket{0^{a}}\ket{\bm b}$, 
and therefore can be estimated by one call to an algorithm for B-EVHM. This completes
the reduction of B-SLEP to B-EVHM.

To design our algorithm for solving S-EVHM, we use a known procedure for
generating a block encoding of $M$  
using the sparse-access oracles for $M$ (Lemma~\ref{lem:block encoding sparse access}).
This block encoding of $M$ is used as an input to our algorithm for B-EVHM, 
which completes the design of our algorithm for S-EVHM.
This concludes the discussion of our approach.

\section{Results}
\label{sec:results}
In this section we prove the query complexity of B-SLEP is $\Theta(\alpha_M/\epsilon)$. In \S \ref{sec:lowerbound} we give a lower bound for B-SLEP. Next, in \S \ref{sec:algorithm} we describe a quantum algorithm that saturates this bound. 
\subsection{Lower bound for B-SLEP}
\label{sec:lowerbound}

In this section, we derive query-complexity lower bounds for B-EVHM, S-EVHM and B-SLEP. We first
derive a lower bound on queries to $U_M$ for S-EVHM, which utilizes existing query-complexity lower bound for AM~\cite{NW99}. Next, we use Lemma \ref{lem:block encoding sparse access} to reduce  S-EVHM to B-EVHM and prove a lower bound for B-EVHM.
We extend this lower bound to B-SLEP by reducing B-EVHM to instances of B-SLEP with $A = \mathds 1$.

We now derive a lower bound
for S-EVHM by a reduction of AM to S-EVHM. En route to this reduction, we define a scaled version of AM and 
obtain the dependence of its query complexity on the scaling factor and
the accuracy.
\begin{problem}[Scaled Approximate Mean (SAM)] \label{pro:SAM}
Given $N \in \mathbb{Z}^+$, $\epsilon \in [\beta/2N,\beta)$, $\beta \in \mathbb{R}^+$ and access to an oracle $Q_g$ encoding a function 
$g:[N] \mapsto [0,\beta]$, return an approximation $\tilde \mu_g \in [0,\beta]$ of the mean $\mu_g := \left(\sum_{j=0}^{N-1} g(j)\right)/N$
such that $|\tilde \mu_g-\mu_g| <  \epsilon $ with probability at least  $2/3$.
\end{problem}

\begin{lemma}\label{lem:SAM alg implies AM alg}
Given an algorithm $\mathcal Q$ that solves SAM for any valid input $(N',1/\epsilon', \beta' , O_g)$ 
by making $K$ queries to~$O_g$, an algorithm $\mathcal R$ can be constructed that solves AM for input  
$(N',\beta'/\epsilon', O_f)$ by making $K$ queries to $O_f$.
\end{lemma}
\begin{proof}
Using queries to $\mathcal Q$, we design a  algorithm $\mathcal{R}$
to solve the AM instance given by the input in the lemma statement. 
Using the input oracle $O_f$, construct a new oracle $O_g$ encoding the function
$g:[N]\mapsto [0,\beta']: j \mapsto \beta'f(j)$. The algorithm 
$\mathcal R$ first obtains a number $\tilde \mu_g$ by running the algorithm $\mathcal{Q}$ with input $(N',1/\epsilon', \beta',O_g)$, and then returns $\tilde \mu_f = \tilde \mu_g/\beta'$. 

The number $\tilde \mu_g$ satisfies $|\tilde \mu_g-\mu_g| < \epsilon' $ with probability at least $2/3$,
where $\mu_g$ is the mean of $g$ as defined in Problem~\ref{pro:SAM}.
Therefore, with probability at least $2/3$, $\tilde \mu_f$ satisfies $|\tilde \mu_f-\mu_f| < \epsilon'/\beta'$ 
as required. Furthermore, each call to $O_g$ can be implemented with one call to $O_f$, therefore
$\mathcal R$ makes $K$ queries to $O_f$.
\end{proof}
We prove the complexity of AM implies a lower bound for SAM.
\begin{lemma} \label{lem:querySAM}
The query complexity of SAM is $\Omega(\beta/\epsilon)$.
\end{lemma}
\begin{proof}
From Corollary \ref{lem:queryAM}, there exist constants $\eta_0, C_0 \in \mathbb{R}^+$ 
such that any quantum algorithm solving AM on input
$(N, 1/\epsilon, O_f)$, satisfying $N,1/\epsilon > \eta_0$,
makes at least $C_0/\epsilon$ queries to $O_f$.
Our proof is by contradiction, so assume SAM has query complexity not in $\Omega(\beta/\epsilon)$. 
Then, there exists a quantum algorithm $\mathcal{Q}$ that solves SAM for some input $(N',
1/\epsilon', \beta', O_g)$, with  $N',\beta',1/\epsilon' > \max(\eta_0, 1)$, using fewer than 
$\lfloor C_0 \beta'/\epsilon'\rfloor$ queries to $O_g$. By Lemma~\ref{lem:SAM alg implies AM alg}, there exists an algorithm $\mathcal R$ solving AM on input $(N', \beta'/\epsilon',O_g)$ which makes fewer than $\lfloor C_0 \beta'/\epsilon'\rfloor = 
\lfloor C_0/\epsilon'\rfloor$ queries to $O_g$. As $N',\beta'/\epsilon', > \eta_0$, 
we have a contradiction to the first statement of the proof.
\end{proof}
To prove the lower bound on the query complexity of S-EVHM, we make use of  an encoding of a function into a sparse Hermitian matrix. 
We first define a function $\operatorname{ind}$ that maps each pair of a row and a column index of a $2^n\times 2^n$
Hermitian matrix $M$ to an integer in $[2^{2n}]$.
\begin{definition}
\label{def:ind}
For $n \in \mathbb{Z}^+$, we define the index function
\begin{equation}
\label{eq:ind}
\begin{aligned}
   {\rm ind}: [2^n]\times[2^n] &\to \left[2^{2n}\right]
    \\
   (i,j) &\mapsto 2^n\big ((j-i)\bmod{2^n}\big ) + i.
\end{aligned}
\end{equation}
\end{definition}
\noindent Here $\bmod {2^n}:\mathbb{Z} \to [2^n]$ is the modulo operation. For example, for $n=2$, the $2^n\times 2^n$ matrix $B$ defined by $B_{ij} = \text{ind}(i,j)$ takes the form
\begin{equation}
    B = \begin{pmatrix}
    0  &  4  &  8  & 12  \\
    13  &  1  &  5  &  9  \\
    10  &  14  &  2  &  6  \\
    7  &  11  &  15  &  3  
    \end{pmatrix}.
\end{equation}
We prove another lemma before explaining our encoding.
\begin{lemma}
\label{lem:indproperty}
For $n \in \mathbb{Z}^+$ 
\begin{enumerate}
    \item ${\rm ind}$ is invertible.
    \item for $i,j \in [2^n]$ such that $i\ne j$,
\begin{equation}
    {\rm ind}(i,j) + {\rm ind}(j,i) \ge 2^{2n}.
\end{equation}
\end{enumerate}
\end{lemma}
\begin{proof}
We prove the two statements separately. 
\begin{enumerate}
    \item 
The function ${\rm ind}$ is invertible if ${\rm ind}(i,j) = {\rm ind}(k,l) $ implies $ i=k, j=l$, which is 
equivalent to 
\begin{equation}
    2^n\left((j-i) \bmod{2^n} - (l-k)\bmod{2^n}\right) = k-i.
\end{equation}
Comparing the left- and right-hand sides of this equation together with the constraint $\abs{k-i} \le 2^n$
readily yields $i=k$ and $j=l$
\item By definition of the ${\rm ind}$ function, 
\begin{equation}
    {\rm ind}(i,j) + {\rm ind}(j,i) = 2^n\big ((j-i)\bmod{2^n} + (i-j)\bmod{2^n}\big ) + i+j
\end{equation}
The result follows readily from the observation 
$((j-i)\bmod{2^n} + (i-j)\bmod{2^n}\big ) = 2^n$ for $i \ne j$.
\end{enumerate}\end{proof}

We are now ready to state the encoding of a function into a matrix that we use for our reduction. 
\begin{definition}
\label{def:encodinggtoA}
For $n \in \mathbb{Z}^+$, $d_M \in [2^n]$, $\beta \in \mathbb{R}^+$ and a function 
$g:\left[d_M2^{n-1}\right] \to [0,\beta]$,
the $(n,d_M)$-matrix encoding of $g$ is the $2^n\times 2^n$ matrix $M$ with entries
\begin{equation} \label{eq:A encoding}
\arraycolsep=5pt\def\arraystretch{1.4}
    M_{ij} := \left\{ \begin{array}{lcl} g({\rm ind}(i,i)) & \text{if} & 
    {\rm ind}(i,i) \in \left[d_M2^{n-1}\right] \; \text{and} \;\; i=j\\
    g({\rm ind}(i,j))/2 & \text{if} & {\rm ind}(i,j) \in \left[d_M2^{n-1}\right] \; \text{and} \;\; i \ne j\\
    g({\rm ind}(j,i))/2 & \text{if} & {\rm ind}(j,i) \in \left[d_M2^{n-1}\right] \; \text{and} \;\; i \ne j\\
    0 &  \text{if} & {\rm ind}(i,j),{\rm ind}(j,i)  \notin \left[d_M2^{n-1}\right] \; \text{and} \;\; i \ne j.
    \end{array}\right.
\end{equation}
\end{definition}
Lemma~\ref{lem:indproperty} implies that for any $i \neq j$, both ${\rm ind}(i,j)$ and ${\rm ind}(j,i)$ cannot simultaneously
belong to $\left[d_M2^{n-1}\right]$, so each entry of $M$ is unique.
We now derive relevant properties of this matrix encoding.
\begin{lemma}
\label{lem:matrixencodingproperty}
For $n \in \mathbb{Z}^+$, $d_M \in [2^n]$, $\beta \in \mathbb{R}^+$ and a function 
$g:\left[d_M2^{n-1}\right] \to [0,\beta]$, let $2^n \times 2^n$ $d_M$-sparse matrix $M$ be the
$(n,d_M)$-matrix encoding of $g$. Then
\begin{enumerate}
    \item $M$ is a Hermitian matrix,
    \item $M$ is $d_M$-sparse, and 
    \item the entries of $M$ satisfy
    \begin{equation}
        \sum_{i,j \in [2^{n}]} M_{ij} =\sum_{k \in [d_M2^{n-1}]} g(k).
    \end{equation}
\end{enumerate}
\end{lemma}
\begin{proof}
We provide a separate proof for each of the three statements. 
\begin{enumerate}
    \item First we prove that $M$ is Hermitian. As all entries of $M$ are real, we only need to prove that
    $M_{ij} = M_{ji}$ for $i \ne j$. Suppose $i,j$ are such that ${\rm ind}(i,j) \in \left[d_M2^{n-1}\right]$.
    Then ${\rm ind}(j,i) \notin \left[d_M2^{n-1}\right]$ by Lemma \ref{lem:indproperty}. Therefore by Eq.~\eqref{eq:A encoding},
    $M_{ij} = M_{ji} = g({\rm ind}(i,j))/2$. A similar argument yields $M_{ij} = M_{ji}$ if 
    ${\rm ind}(i,j) \notin \left[d_M2^{n-1}\right]$.  
    \item To prove that $M$ is $d_M$-sparse, we count, for each row $i$, the number of column indices $\{j\}$ 
   for which $M_{ij}$ is non-zero. For an entry $M_{ij}$ to be non-zero, either ${\rm ind}(i,j)$
    or ${\rm ind}(j,i)$ must belong to $[d_M2^{n-1}]$. We first count the number of $\{j\}$ such that ${\rm ind}(i,j) \in [d_M2^{n-1}]$.
    Note that for $(j-i) \mod{2^n} > d_M/2$, we have
    ${\rm ind}(i,j)>d_M2^{n-1}$, therefore there are at most
    $\lceil d_M/2\rceil$ values of $j$ satisfying ${\rm ind}(i,j) \in [d_M2^{n-1}]$ for any $i$. Similar argument can be made
    to show that there are at most $\lceil d_M/2\rceil$ values of $j$ for which ${\rm ind}(j,i) \in [d_M2^{n-1}]$. Considering
    that $j=i$ features in both these lists, we conclude that any row $i$ has at most $d_M$ non-zero entries.
    \item This property is a consequence of the invertibility of ${\rm ind}$ and the Hermiticity of $M$. Every $k \in [d_M2^{n-1}]$
    in the domain of $g$ is mapped to a pair of a row and a column index $(i,j)$ by the inverse of the ${\rm ind}$ function. 
    Furthermore, if $i \ne j$, then $M_{ij} = g(k)/2$ so that $M_{ij} + M_{ji} = g(k)$ by Hermiticity of $M$.
\end{enumerate}
\end{proof}
\noindent The $(n,d_M)$-matrix encoding $M$ of $g$
can be constructed with a single query to the function oracle, as given by the following lemma. 
\begin{lemma}
\label{lem:oraclegtoA}
Given $n$ qubits, a sparsity parameter $d_M \in [2^n]$, a number $\beta_M \in \mathbb{R}^+$ and an oracle $O_g$ encoding a function 
$g:\left[d_M2^{n-1}\right] \to [0,\beta]$, sparse access $(d_M, \beta_M, O_{M_{\rm val}},O_{M_{\rm loc}})$ to the $(n,d_M)$-matrix encoding of $g$ can be constructed with  $O_{M_{\rm val}}$ making a single query to $O_g$ and $O_{M_{\rm loc}}$ making zero queries to $O_g$.
\end{lemma}
\begin{proof} The oracle for $M_{\rm val}$  
can be constructed using the explicit formula  
\eqref{eq:A encoding} for the entries of $M$. 
This construction requires first computing ${\rm ind}(i,j)$, then computing $g({\rm ind}(i,j))$
by making one query to $O_g$,
and finally uncomputing ${\rm ind}(i,j)$. The function ${\rm ind}(i,j)$ can be computed
using the explicit formula \eqref{eq:ind} and requires no queries to $O_g$.

To construct an oracle for $M_{\rm loc}$, it suffices to provide an explicit formula for 
computing $M_{\rm loc}$. Recall that
\begin{equation}
\begin{aligned}
    M_{\rm loc}:[2^n]\times [d_M] &\to [2^n]\\(i,l)&\mapsto M_{\rm loc}(i,l),
\end{aligned}
\end{equation}
where $M_{\rm loc}(i,l)$ is the $l$th non-zero entry in the $i$th row of $M$.
From the proof of the second statement of Lemma~\ref{lem:matrixencodingproperty}, we deduce that $M_{\rm loc}$ is given by 
$M_{\rm loc}(j,l) = i+l-\lfloor d_M/2 \rfloor \bmod{2^n}$. 
\end{proof}

The function to matrix encoding in Def.\ \ref{def:encodinggtoA} ensures that the mean of a function $g$
can be inferred from the expectation value of the $(n,d_M)$-matrix encoding of $g$ with respect to the equal superposition state, as stated in the following lemma. 
\begin{lemma}
\label{lem:meangtoA}
Let $n \in \mathbb{Z}^+$, $d_M \in [2^n]$, $\beta \in \mathbb{R}^+$, 
$g:[d_M2^{n-1}] \to [0,\beta]$ and let $M$ be the $(n,d_M)$-matrix encoding of~$g$. Then the mean $\mu_g$ of $g$ satisfies
\begin{equation}
    \mu_g = \frac{2}{d_M}\braket{+^{n} | M | +^{n}},
\end{equation}
where $\ket{+^{n}} = H^{n}\ket{0^n}$.
\end{lemma}

\begin{proof}
By Lemma~\ref{lem:matrixencodingproperty}, $\sum_{i,j \in [2^n]}M_{ij} = \sum_{i\in \left[d_M2^{n-1}\right]}g(i)$. Therefore,
\begin{align}
    \braket{+^{n} | M | +^{n}} &= \frac{1}{2^n}\sum_{i,j \in [2^{n}]} M_{ij} 
    \nonumber\\&= \frac{1}{2^n}\sum_{i \in [d_M2^{n}]} g(i)
    \nonumber\\ &= \frac{d_M2^{n-1}}{2^n}\mu_g \\&= \frac{d_M}{2}\mu_g,
\end{align}
where the second equality follows from Lemma~\ref{lem:matrixencodingproperty}.
\end{proof}

We now give our reduction from SAM to S-EVHM.
\begin{lemma}\label{lem:EVFSO alg implies SAM alg}
Given an algorithm $\mathcal Q$ that solves the problem S-EVHM on any instance 
\begin{equation} \label{eq:EVFSO input for lower bound proof}
(n, 1/\epsilon, (d_{M}, \beta_{M}, O_{M_{ \rm val}},O_{M_{ \rm loc}}), H^{\otimes n})
\end{equation}
by making $K$ queries to the sparse-access oracles encoding $M$,
an algorithm $\mathcal R$ can be constructed that solves the SAM instance with input  
\begin{equation} \label{eq:SAM input for lower bound proof}
(d_{M}2^{n-1}, d_{M}/2\epsilon, \beta_{M}, O_g),
\end{equation}
by making $K$ queries to $O_g$.
\end{lemma}
\begin{proof}
We show how to construct an algorithm $\mathcal R$ as described in the Lemma.  
Let $M$ be the $(n,d_M)$-matrix encoding of $g$, which is given by Def. \ref{def:encodinggtoA} and has sparse-access oracles given in  Lemma \ref{lem:oraclegtoA}. Let $\mathcal R$ be the algorithm that queries $\mathcal Q$ with input given by Eq.~\ref{eq:EVFSO input for lower bound proof} to compute $\tilde{u}$ and 
returns $2\tilde{u}/d_{M}$. 
By Lemma~\ref{lem:meangtoA},
$\braket{+^{n} | M | +^{n}} = d_{M}\mu_g/2$.
Then 
$\left|\tilde{u} - d_{M}\mu_g/2\right| < \epsilon$
with probability at least $2/3$,
and so
$\left|2\tilde{u}/d_{M} - \mu_g \right| <  2\epsilon/d_{M}$ 
with probability $2/3$.
Therefore, $\mathcal R$ solves the SAM instance defined by the inputs in Eq.~\eqref{eq:SAM input for lower bound proof}.
Also, $\mathcal R$ makes $K$ queries to $O_g$ as $\mathcal Q$ makes~$K$ queries to oracles encoding $M$, and each query to $M$
requires only one query to $O_g$ by Lemma~\ref{lem:oraclegtoA}.
\end{proof}

Using this reduction, we derive the following query-complexity lower bound. 
\begin{lemma}

\label{thm:SLE2queryB}
Given $n$ qubits, sparse access $(d_M, \beta_M, O_{M_{\rm val}},O_{M_{\rm loc}})$  to a $2^n \times 2^n$ Hermitian matrix $M$, an accuracy $\epsilon \in (d_M\beta/2^n,d_M\beta)$, and an $n$-qubit unitary black box $V$,
any quantum algorithm that returns with probability at least $2/3$ an $\epsilon$-additive approximation 
of $\braket{0|V^\dagger M V|0}$ for $V = H^{\otimes n}$ makes $\Omega(d_M\beta/\epsilon)$ queries $U_M$. 
\end{lemma}

\begin{proof}
From Lemma \ref{lem:querySAM}, there exist constants $C_0, \eta_0 \in \mathbb{R}^+$
such that any quantum algorithm solving the SAM problem with input 
$(N, 1/\epsilon, \beta, O_g)$, where $N,1/\epsilon, \beta > \eta_0$,
makes at least $\lfloor C_0 \beta/\epsilon \rfloor$ queries to $O_g$. We proceed by contradiction.
Assume Lemma~\ref{thm:SLE2queryB} is false, then there exists a quantum algorithm $\mathcal{Q}$ which for some input 
\begin{equation}
    (n',  (d'_M, \beta'_M, O_{M_{\rm val}},O_{M_{ \rm loc}}), 1/\epsilon')
\end{equation}
with  $n', d'_M,1/\epsilon',\beta'_M > 2\eta_0$,
returns an $\epsilon$-additive approximation of $\braket{0|V^\dagger M V|0}$ using a number of queries to oracles encoding $M$ fewer than $\lfloor C_0d'_M\beta'_M/2\epsilon' \rfloor$.
Let $N'' := d'_M2^{n'-1}, 1/\epsilon'' := d'_{M}/2\epsilon_0$ and $\beta'' := \beta'_{M}$.
Then by Lemma~\ref{lem:EVFSO alg implies SAM alg} there exists an algorithm $\mathcal R$ solving SAM on input 
\begin{equation}
    (N'' ,1/\epsilon'', \beta'', O_g),
\end{equation}
which makes fewer than $\lfloor C_0d'_M\beta'_M/2\epsilon' \rfloor= \lfloor C_0\beta''/\epsilon'' \rfloor$ queries  to $O_g$. 
Using $n', d'_M,1/\epsilon',\beta'_M > 2\eta_0$ we have $N'', 1/\epsilon'', \beta'' > \eta_0$. Therefore, the existence of $\mathcal R$ contradicts the first statement of the proof.
\end{proof}

Lemma \ref{thm:SLE2queryB} addresses S-EVHM under the restriction $V = H^{\otimes n}$. This lemma implies the following lower bound for S-EVHM. 
\begin{coro} \label{col:LB for S-EVHM}
Any algorithm that solves S-EVHM makes $\Omega(d_M \beta_M/\epsilon)$ queries to sparse-access oracles for $M$. 
\end{coro}
\noindent We next use Corollary \ref{col:LB for S-EVHM} to prove a similar query-complexity lower bound for B-EVHM. 
\begin{thm}
\label{thm:queryB-EVHM}
Any algorithm that solves B-EVHM makes $\Omega(\alpha_M/\epsilon)$ queries to the block encoding of $M$.
\end{thm}
\begin{proof}
From Theorem \ref{thm:SLE2queryB}, there exist constants $\eta_0, C_0 \in \mathbb{Z}^+$ 
such that any quantum algorithm solving the problem S-EVHM on input 
$(n,(d_M,\beta_M, O_{M_{\rm val}},O_{M_{\rm loc}}),1/\epsilon,V)$, satisfying $n,d_M,\beta_M,1/\epsilon > \eta_0$   
makes at least $\lfloor C_0d_M\beta_M/\epsilon \rfloor$ queries to $M$.
Our proof is  by contradiction. Assume that a quantum algorithm $\mathcal Q$ solves 
B-EVHM with queries to $U_M$ 
not in $\Omega(\alpha_M/\epsilon)$. 
Then there exists an input $(n', (\alpha'_M, a'_{M}, 0, U_{M}), 2/\epsilon', V)$, 
with $n', \alpha'_{M}, a'_M, 1/\epsilon' > (\eta_0+1)^2$ and $a'_M \geq 2$, for which $\mathcal{Q}$ solves B-EVHM using fewer than 
$\lfloor C_0 \alpha'_{M}/\epsilon' \rfloor$ 
queries to $U_M$. For 
\begin{equation} d'_{M} = \eta_0+1\text{ and }\beta'_M = \alpha'_{M}/(\eta_0 +1), \end{equation} and 
following Lemma \ref{lem:block encoding sparse access}, block access $(d'_{M} \beta'_M,a'_M,\epsilon'/2, U_M)$ can be constructed such that $U_M$ 
makes a single query to sparse-access oracles for $M$.
Applying $\mathcal Q$ on the instance defined by input $(n', (d'_M \beta'_{M},\alpha'_M,0, U_M), 2/\epsilon',V)$ gives a quantum algorithm that solves S-EVHM with input
$(n', (d'_{M}, \beta'_{M},  O_{M_{\rm val}},O_{M_{\rm loc}} ), 1/\epsilon',V)$ by making 
fewer than $\lfloor C_0d'_M\beta'_M/\epsilon'\rfloor = \lfloor C_0 \alpha_M/\epsilon_0 \rfloor $
queries to $M$. As $n', d'_{M}, \beta'_M, 1/\epsilon' > \eta_0$ in the input to this instance of S-EVHM, 
we have a contradiction to the first statement of the proof.
\end{proof}

Theorem~\ref{thm:queryB-EVHM} establishes a query complexity lower bound for the problem B-EVHM, in which the matrix $M$ is given via block encoding, and the state is generated by a unitary black box~$V$. 
In Appendix \ref{sec:appendix}, 
we show that the lower bound in Theorem~\ref{thm:queryB-EVHM} also applies for the problem of estimating 
the expectation value even if $V$ is any fixed 
unitary operator and not a black box. In particular for the restriction $V=\mathds{1}_n$, the output of this problem
is an estimate  of $\braket{0^n|M|0^n}$. Interestingly, the restriction of S-EVHM
to $V=\mathds{1}_n$ can be solved by making just one query to the sparse oracles for~$M$.

As B-EVHM is a restriction of B-SLEP to $A=\mathds{1}_n$, 
the lower bound in Theorem~\ref{thm:queryB-EVHM} immediately yields a lower bound for B-SLEP.
\begin{thm}
\label{coro:SLE2queryB}
Any quantum algorithm that solves B-SLEP must make $\Omega(\alpha_M /\epsilon)$
queries to $U_M$.
\end{thm}
\noindent This theorem concludes our derivation of lower bound for B-SLEP.

\subsection{Tightness of  bound}
\label{sec:algorithm}
\label{subsubsec:B-SLEP}

In this section, we construct quantum algorithms for solving B-SLEP that achieves the query-complexity
lower bounds derived in the previous section. 
Our algorithm for B-SLEP follows the approach outlined in \S\ref{sec:approachreduction}. 
We use the algorithm for B-EVHM discussed in \S\ref{sec:backgroundbevhm} 
as a subroutine to design an algorithm for B-SLEP in~\S\ref{subsubsec:B-SLEP}. 
We first show how to construct a block encoding of 
$A^{-1}$ using techniques discussed in \S\ref{sec:backgroundlowerbound}. Our algorithm for B-SLEP 
makes one query to our algorithm for B-EVHM (Algorithm~\ref{alg:AB-EVHM}) and queries the block encoding of
$A^{-1}$ as an input to this single query.

First we combine Theorem \ref{thm:P(H)} and Corollary \ref{thm:poly approx of 1/x} to show how block access to $A^{-1}$ can be constructed using block access to $A$. A more general result in this direction is given in Ref.~\cite{CGJ19}. We provide a different 
derivation in our setting and report complexities for queries to block encoding of $A$ and 2-qubit gates separately.
\begin{coro} \label{coro:A-1}
Given $n$ qubits, an accuracy $\epsilon \in (0,2]$, a $\kappa \ge 1$, and block access $(\alpha_A, a_A, 0, U_A)$ to an invertible $2^n\times 2^n$ Hermitian matrix $A$ such that $\norm{A^{-1}}\alpha_A \le  \kappa$, then block access $(8\kappa/3, a_A+2, \epsilon, U_{A^{-1}})$ to  $A^{-1}$ can be constructed such that $U_{A^{-1}}$ makes
$\mathcal  O\left(\kappa\log(\kappa /\epsilon)\right)$
queries to $U_A$  and
$\mathcal O\left(\kappa a_A\log(\kappa /\epsilon)\right)$
additional 2-qubit gates.  

\end{coro}
\begin{proof}
Invoking Corollary \ref{thm:poly approx of 1/x} with the theorem's $\delta$ parameter equal to $1/\kappa$ and the theorem's $\epsilon$ parameter equal to $3\epsilon/(8\kappa)$,
there exists a polynomial $\mathcal P$ with degree in $\mathcal O(\kappa\log(\kappa/\epsilon))$ such that for all 
\begin{equation}
   x \in  [-1,1] \setminus \left[-1/\kappa, 1/\kappa\right],
\end{equation}
$|\mathcal P(x)/2| \leq 1/2$ and 
\begin{equation}
    \left| \frac{\mathcal P(x)}{2} - \frac{3}{8\kappa } \frac{1}{x}\right| < \frac{3\epsilon}{16\kappa}.
\end{equation} 
As the spectrum of $A/\alpha_A$ lies in $[1,-1] \setminus \left[-1/\kappa, 1/\kappa \right]$, 
each eigenvector $\ket{\lambda}$ of $A$ satisfies
\begin{equation}
   \abs{\Big\langle \lambda \Big|\left (\frac{3}{8\kappa} A^{-1} - \frac{\mathcal P(A/\alpha_A)} 2\right)\Big|\lambda \Big \rangle} \leq \frac{ 3\epsilon}{16\kappa},
\end{equation}
and consequently
\begin{equation}
\label{eq:blockA^{-1}}
    \left\|\frac{3}{8\kappa } A^{-1} -\frac{ \mathcal P(A/\alpha_A)}{2}\right \| \leq \frac{ 3\epsilon}{16\kappa}. 
\end{equation}
Next, Theorem \ref{thm:P(H)} with $\sigma = 3\epsilon/16\kappa$ guarantees that a $(1, a_A+2, 3\epsilon/16\kappa )$ 
block encoding of $\mathcal P(A/\alpha_A)/2$ can be constructed using $\mathcal  O\left( \kappa \log(\kappa /\epsilon)\right)$
queries to $U_A$  and $\mathcal   O\left(\kappa a_A\log(\kappa /\epsilon)\right)$
additional 2-qubit gates. From Eq.~\eqref{eq:blockA^{-1}}, this block encoding of $\mathcal P(A/\alpha_A)/2$ also block encodes
$A^{-1}$ with the prefactor $8\kappa/3$ and error bounded by
\begin{equation}
    (8\kappa/3)(3\epsilon/16\kappa + 3\epsilon/16\kappa) = \epsilon.
\end{equation}
Therefore, the block encoding of $\mathcal P(A/\alpha_A)/2$ given by Theorem \ref{thm:P(H)} is the required block access
$(8\kappa/3, a_A+2,\epsilon, U_{A^{-1}})$  to $A^{-1}$.
\end{proof}

We next prove that block access to $M$ can be used to construct block access to $M^{(m)}:= M \otimes \ket{0^m} \bra{0^m}$
for any given $m \in \mathbb{Z}^+$.
\begin{lemma}\label{lem:block encoding c-H}
Given $n$-qubits, block access $(\alpha_M,a_M,0,U_M)$ to a $2^n \times 2^n$ Hermitian matrix $M$, and given an $m \in \mathbb{Z}^+$, 
then block access $(\alpha_{M},a_M+1,0,U_{M^{(m)}})$ to $M^{(m)}$  can be constructed such that $U_{M^{(m)}}$ queries $U_M$ once and 
employs $\mathcal O(m^2)$ 
additional 2-qubit gates. 
\end{lemma} 
\begin{proof}
We show that the desired block encoding of~$M^{(m)}$ is given by the $(1 + a_M + n + m)$-qubit circuit 
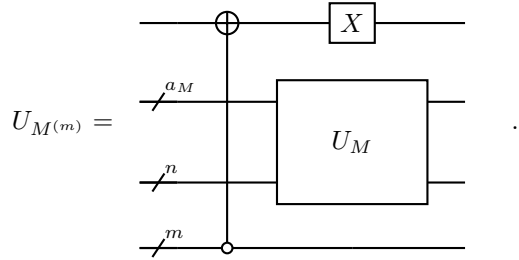
\begin{center}
$U_{M^{(m)}} =$ \begin{quantikz}
 &\qw &\targ{}& \gate{X} & \qw &\\
& \qw \qwbundle{a_M}& \qw &\gate[wires = 2][2cm]{U_M} &\qw &\\
 & \qw \qwbundle{n} &\qw &  &\qw &\\
 &\qw \qwbundle{m}& \octrl{-3}&\qw &\qw &
\end{quantikz}.
\end{center}
The CNOT gate in this circuit has $m$ control qubits, and only flips the target qubit if the control qubits are in the state $\ket{0^m}$. For any computational basis states $\ket{0^{a_M+1},s_1,0^m}$ and 
  $\ket{0^{a_M+1},s_2,0^m}$, where $s_1$ and $s_2$ are
  $n$-bit strings, 
\begin{equation}
\label{eq:M^(m)1}
   \bra{0^{a_M+1},s_1,0^m}U_{M^{(m)}}\ket{0^{a_M+1},s_2,0^m} = M_{s_1s_2}/\alpha_M
\end{equation}
is true. Furthermore, for any states $\ket{0^{a_M+1},s_1,r_1}$ and $\ket{0^{a_M+1},s_2,r_2}$, where $r_1$ and $r_2$ are $m$-bit strings such that at least one of $r_1$ and $r_2$ is not equal to $0^m$, 
\begin{equation}
\label{eq:M^(m)2}
   \bra{0^{a_M+1},s_1,r_1}U_{M^{(m)}}\ket{0^{a_M+1},s_2,r_2} = 0.
\end{equation}
Equations~\eqref{eq:M^(m)1} and \eqref{eq:M^(m)2} together imply
\begin{equation}
    \bra{0^{a_M+1}}U_{M^{(m)}}\ket{0^{a_M+1}} =   (M \otimes \ket{0^m}\bra{0^m})/\alpha_M,
\end{equation}
as required. 
The required number of 2-qubit gates scales as $\mathcal{O}(m^2)$, which 
follows from the fact that
the $0^{m}$-CNOT gate can be implemented using $\mathcal O(m)$ 2-qubit gates
and the Toffoli gate~\cite{NC}. 
\end{proof}

\begin{thm} 
\label{thm:SLE--$f$}
A quantum algorithm can be constructed that takes $n$ qubits, 
a $\kappa \ge 1$,
block access $(\alpha_A,a_A,0, U_A)$ to a $2^n\times 2^n$ invertible Hermitian matrix $A$ 
such that $\norm{A^{-1}}\alpha_A \le  \kappa$, 
block access $(\alpha_M,a_M,0, U_M)$ to a $2^n\times 2^n$ Hermitian matrix $M$, 
an accuracy $\epsilon \in [\alpha_M/2^n,\alpha_M]$,
and an $n$-qubit unitary black box $U_{\bm b}$, and returns with probability at least $2/3$ an $\epsilon$-additive approximation to $\bm x^\dagger M \bm x$, 
where $\bm{x} := A^{-1}{\bm{b}}$ and $\bm{b}$ is the $2^n$-dimensional vector with entries $b_i =\braket{i| U_{\bm b}|0}$, by making
$
\mathcal O(\alpha_M\kappa^2/\epsilon)
$
queries to $U_M$, $
 \mathcal O\left(\alpha_M  \kappa^3 \log\left(\alpha_M\kappa^2/\epsilon\right)/\epsilon\right)$
queries to  $U_A$, $
\mathcal O(\alpha_M\kappa^2/\epsilon)
$ queries to $U_{\bm b}$, and  
\begin{equation}\label{eq:final 2 qubit complexity}
    \mathcal O\left(\frac{\alpha_M \kappa^2  }{\epsilon} \left( n + a_M + a_A \kappa\log\left(\frac{ \alpha_M \kappa^2}{\epsilon}\right) + a_A \right)\right),
 \end{equation} 
additional 2-qubit gates. 
\end{thm}

\begin{proof}
We describe an algorithm (Alg.\ \ref{alg:AB-SLEP}) for B-SLEP, and  prove the correctness 
of this algorithm. This algorithm makes reference to subroutines described in other work.
We end the proof by analyzing the algorithm's complexity.
The algorithm is as follows.
\begin{algorithm}[H]
\begin{algorithmic}[1]
\caption{Algorithm for Block-access System of Linear Equations Problem \label{alg:AB-SLEP}}
\Require{
$(n,\kappa,(\alpha_A,a_A,0,U_A),(\alpha_M,a_M,0,U_M),\epsilon,U_{\bm b})$ }
\Ensure{
With probability at least $2/3$ an $\epsilon$-additive approximation to $\bm x^\dagger M \bm x$, 
where $\bm{x} := A^{-1}{\bm{b}}$ and $\bm{b}$ is the vector with entries $b_i = \braket{i| U_{\bm b}|0}$
}

\hspace{-1.3cm}\ \textbf{Procedure:}
    \State
    \label{line:gamma}
$
        \gamma = \begin{cases} \alpha_M\text{ if }\alpha_M \geq1\\
           \sqrt{\alpha_M}\text{ if }\alpha_M < 1.
           \end{cases}
  $
    \State Construct  block access $(8\kappa/3, a_A+2, \epsilon/8\gamma\kappa, U_{A^{-1}} )$ to $A^{-1}$.
    \Comment{Use Corollary \ref{coro:A-1} }
    \State  Construct block access $(\alpha_{M}, a_{M}+1, 0, U_{ M^{(a_A+2)}})$  
    to $M^{(a_A+2)}:= \ket{0^{a_A+2}}\bra{0^{a_A+2}} \otimes M$.
    \Comment{Use Lemma~\ref{lem:block encoding c-H}}
    \State Compute $\tilde r$ by executing the algorithm for B-EVHM with input \Comment{Use Algorithm \ref{alg:AB-EVHM}}
    $$\left(n +a_A+2,\ (\alpha_M,  a_M+1, 0,  U_{ M^{(a_A+2)}}),\  (\epsilon/2)(8\kappa/3)^{-2},\  
    \left( U_{A^{-1}}\right)(\mathds{1}_{{a_A+2}}\otimes U_{\bm b})\right).$$
    \State \Return $\tilde{u}=(8\kappa/3)^2\tilde{r}$.
\end{algorithmic}
\end{algorithm}
\begin{proof}[Correctness]
Let 
\begin{align}
\label{eq:final SLEP output}
   r := \bra{0^{n +a_A+2}}\left(\mathds{1}_{a_A+2}\otimes U_{\bm b}^\dagger\right) \left(U_{A^{-1}}^\dagger \right)  
    \Big(\ket{0^{a_A+2}}\bra{0^{a_A+2}} \otimes M\Big )  \Big( U_{A^{-1}}\Big)\Big(\mathds{1}_{{a_A+2 }}\otimes U_{\bm b}\Big)\ket{0^{n +a_A+2}}.
\end{align}
The output of B-EVHM (Algorithm \ref{alg:AB-EVHM}) satisfies
\begin{equation}
    |r - \tilde{r}| \le (\epsilon/2)(8\kappa/3)^{-2},
\end{equation}
with probability at least $2/3$.
Let $u := \bra{\bm x}  M  \ket{\bm x} $, where $\ket {\bm x}$ is the solution state to~$A\bm x = \bm b$, which can be expressed as 
\begin{align}
    u = \bra{\bm b}  {A^{-1}}^\dagger M A^{-1}  \ket{\bm b}
    = \bra{0^{n}}U_{\bm b}^\dagger {A^{-1}}^\dagger M A^{-1}U_{\bm b}\ket{0^{n}}.
\end{align}
Define $\tilde{E} :=  (8\kappa/3)\bra{0^{a_A + 2}}  U_{A^{-1}} \ket{0^{a_A + 2}}$. Then, by the definition of block encoding (Def.~\ref{def:block encoding}), we have
\begin{equation}
    { \left\| A^{-1}  - \tilde{E} \right\| \leq \epsilon/8\gamma\kappa}.
\end{equation}
Furthermore, this bound gives $
     \norm{\tilde{E}\ket{\bm b} - A^{-1}\ket{\bm b}} \nonumber 
    \le \epsilon/8\gamma\kappa$.
Analyzing Eq.~\eqref{eq:final SLEP output} yields $ (8\kappa/3)^2 r = \bra{\bm b} \tilde E^\dagger M \tilde E\ket{\bm b}$.    
Therefore,
\begin{align}
  \left|u - (8\kappa/3)^2r\right| 
   &= \left| \bra{\bm b}{A^{-1}}^\dagger  M A^{-1}\ket{\bm b} - \bra{\bm b}  \tilde E^\dagger M \tilde E\ket{ {\bm b}}\right| \nonumber \\
   & \le \norm{M} \left( \left|\|\tilde E\ket{ {\bm b}}\|^2 - \|A^{-1}\ket{\bm b}\|^2 \right| \right)  \nonumber\\
   & \le \|M\| \left( \left|\|\tilde E\ket{ {\bm b}}\| - \|A^{-1}\ket{\bm b}\| \right| \right)   \left( \|\tilde E\ket{ {\bm b}}\| + \|A^{-1}\ket{\bm b}\| \right)  \nonumber \\
   & \le \alpha_M \frac{\epsilon}{8\gamma \kappa}  \left(\left(\frac{\epsilon}{8\gamma \kappa} +\kappa\right)+\kappa\right)  \nonumber\\
   & \le   \frac{\alpha_M}{\gamma^2}\frac{\epsilon^2}{64  \kappa^2}   + \frac{\alpha_M}{\gamma}\frac{\epsilon}{4}    \nonumber \\
   & \le \epsilon/2,
\end{align}
where the inequalities $\alpha_M/\gamma, \alpha_M/\gamma^2 \leq 1$ are used in the last step.
By the triangle inequality, 
\begin{align}
    |\tilde{u}-u| &\le |\tilde{u} - (8\kappa/3)^2r| + |u - (8\kappa /3)^2r | \nonumber \\
    & =  |(8\kappa /3)^2\tilde{r} - (8\kappa /3)^2r| + |u - (8\kappa /3)^2r| \nonumber \\
    & \le \epsilon/2 + \epsilon/2 = \epsilon,
\end{align}
with probability at least $2/3$.
\renewcommand{\qedsymbol}{}
\end{proof}

\begin{proof}[Complexity]
{
We consider the instances with $\alpha_M \geq 1$, which are sufficient to determine big-O scaling
for queries and 2-qubit gates used by our algorithm. For these instances, $\gamma = \alpha_M $ by the first line of Algorithm~\ref{alg:AB-SLEP}. 
{ Using the complexities from Theorem \ref{thm:B-EVHM}, the single call made to Algorithm~\ref{alg:AB-EVHM} makes
\begin{equation}
\mathcal    O(\alpha_M2(8\kappa/3)^{2}/\epsilon) =\mathcal O(\alpha_M\kappa^2/\epsilon),
\end{equation}
queries to $U_{A^{-1}}(\mathds{1}_{{a_A+2}}\otimes U_{\bm b})$ and $U_{M^{(a_A+2)}}$.
Therefore, using the complexity of $U_{A^{-1}}$ given by Corollary~\ref{coro:A-1}, the algorithm makes $\mathcal O(\alpha_M\kappa^2/\epsilon)$ 
queries to each of $U_{\bm b}$, $U_M$ and $U_{A^{-1}}$.
Each query to $U_{A^{-1}}$ comprises 
\begin{equation}
   \mathcal O\left({\kappa}\log(8\alpha_M\kappa^2/\epsilon)\right) = \mathcal  O(\kappa\log(\alpha_M\kappa^2/\epsilon)),
\end{equation}
queries to $U_A$.} Thus, our algorithm makes
$\mathcal O\left(\alpha_M  \kappa^3\log\left(\alpha_M\kappa^2/\epsilon\right)/\epsilon\right)$
queries to $U_A$. 
 
Now we count the additional 2-qubit gates used by our algorithm. 
We separately count 2-qubit gates used directly by Algorithm~\ref{alg:AB-EVHM}, 
2-qubit gates used by queries to $U_{A^{-1}}(\mathds{1}_{{a_A+2}}\otimes U_{\bm b}))$, and 2-qubit gates used by queries to $U_{M^{(a_A+2)}}$. 
We then combine these three contributions to yield the total number of 2-qubit gates used. 
First, the number of 2-qubit gates used directly by Algorithm~\ref{alg:AB-EVHM} is 
 \begin{align} \label{eq:2qg EVHM}
     \mathcal O\left(\alpha_M \kappa^2(n + a_A+2 + a_M)/\epsilon\right) =
     \mathcal O\left(\alpha_M \kappa^2(n + a_A + a_M)/\epsilon\right).
\end{align}
Next, $U_{A^{-1}}(\mathds{1}_{{a_A+2}}\otimes U_{\bm b})$ is queried
$\mathcal O(\alpha_M\kappa^2/\epsilon)$ times by B-EVHM. By Corollary~\ref{coro:A-1},  $U_{A^{-1}}$ uses
\begin{equation}
\mathcal O\left(\kappa a_A\log( 8\alpha_M\kappa^2/\epsilon)\right) \in 
\mathcal O\left(\kappa a_A\log( \alpha_M \kappa^2/\epsilon)\right)
\end{equation}
additional 2-qubit gates.
Consequently,
  \begin{equation} \label{eq:2qg U_A}
    \mathcal  O\left(\alpha_M  \kappa^3 a_A\log\left(\alpha_M\kappa^2/\epsilon\right) /\epsilon\right)
 \end{equation}
2-qubit gates are used throughout all queries to $U_{A^{-1}}(\mathds{1}_{{(a_A+2)}}\otimes U_{\bm b}))$. 
Finally, we count the number of 2-qubit gates used in queries to $U_{M^{(a_A+2)}}$. 
By Lemma~\ref{lem:block encoding c-H}, each query to the oracle $U_{M^{(a_A+2)}}$ requires $\mathcal O(a_A)$ gates, and this oracle is queried $\mathcal O(\kappa^2\alpha_M/\epsilon)$ times by B-EVHM, which yields a total of
\begin{equation} \label{eq: 2qg U}
    \mathcal O\left(a_A \alpha_M \kappa^2 /\epsilon\right)
\end{equation}
additional 2-qubit gates used throughout calls to $U_{M^{(a_A+2)}}$. 
Combining these three contributions  (\ref{eq:2qg EVHM}--\ref{eq: 2qg U})
yields the total number of additional 2-qubit gates given by Eq.~(\ref{eq:final 2 qubit complexity}).}
\end{proof}
\renewcommand{\qedsymbol}{}
\end{proof}

{In comparison with algorithms for QLSP \cite{CKS17}, which have linear dependence on $\kappa$, our algorithm has cubic dependence on $\kappa$. This additional dependence on $\kappa$ comes from the added challenge of solving B-SLEP over QLSP. In particular, the proportionality in the encoding of $\bm x$ into the quantum state $\ket {\bm x}$ is linearly dependent on $\kappa$. Thus, the proportionality of $\bra {\bm x} M \ket {\bm x}$ is quadratically dependent on $\kappa$, which increases the accuracy to which $\ket {\bm x}$ must be computed. 
}
We end this section by noting that the complexity of Algorithm~\ref{alg:AB-SLEP} has no dependence on 
$\alpha_A$.
This is merely a consequence
of the inequality
$\kappa \ge \alpha_A\norm{A^{-1}}$. Therefore,
a block encoding of $A$ with large $\alpha_A$ in the input is always accompanied by a proportionally large value of $\kappa$,
thereby indirectly increasing the query- and 2-qubit gate cost of our algorithm.

\section{Discussion}
\label{sec:discussion}
We proved two main results in this paper. First, for the formulation of SLEP in which $M$ is provided via block access, namely B-SLEP,
we established a lower bound on queries to $M$. Second, we constructed a quantum algorithm 
for solving B-SLEP that saturates this lower bound, thereby proving that
the lower bound is tight. We now discuss the implications of these two results.

We begin by comparing our lower bound with two previously known hardness results 
for restrictions of SLEP, as discussed in \S\ref{sec:backgroundlowerbound}.
The first of these results, presented in Theorem \ref{lem:HHLlowerbound2}~\cite{HHL09},
rules out the existence of a quantum algorithm solving SLEP with $\polylog(N,1/\epsilon)$ scaling of the total run-time. 
This result relies on no assumptions about how the input $M$ is given, but instead requires the complexity-theoretic conjecture
$\textbf{BQP} \neq \textbf{PP}$. Our derivation of the lower bound does not make any complexity-theoretic assumption, 
but instead relies on the assumption that block access to $M$ is given. Furthermore, our lower bound establishes a more stringent constraint $\Omega(1/\epsilon)$ on the scaling in~$\epsilon$ of queries to $M$, 
thereby ruling out a
sub-linear scaling in $\epsilon$. 
As state-of-the-art algorithms for generation of $\ket{\bm{x}}$ run in $\polylog(N,1/\epsilon)$ time,
our result shows that the expectation-value-estimation step is exponentially harder
with respect to scaling in $\epsilon$ than the $\ket{\bm{x}}$-generation step.

The second previously known hardness result, namely Theorem~\ref{lem:HHLlowerbound}, states that if $A$ is given by oracles,
then no quantum algorithm for SLEP 
with $\poly(\log N,\kappa)$ dependence can have query complexity in $\mathcal{O}(1/\epsilon)$.
This result only applies to those cases in which $A$ is 
given by oracles, and does not apply, for instance, if $A$ is fixed or if $A$ is given by a description of a quantum circuit. 
In contrast, we establish our lower bound for B-SLEP by first deriving a lower bound for the case in which $A = \mathds{1}_n$ is fixed. 
Therefore, our result shows that linear scaling on $1/\epsilon$ is optimal if either of the inputs $A$ and $M$ is given by oracles.
SLEP for fixed $A$ arises, for instance, in solving Poisson equation using the finite-element method with a 
fixed choice of basis function~\cite{MP2016}.

En route to deriving our lower bound for queries
to B-SLEP, we derive a lower bound for S-EVHM (Problem~\ref{pro:S-EVHM}), which is the problem of estimating the 
expectation value of a sparse-access Hermitian matrix $M$ with respect to the state generated by a given 
unitary black box. 
Our lower bound for S-EVHM is $\Omega(d_M\beta_M/\epsilon)$, where
$d_M$ is the sparsity of $M$ and $\beta_M$ is an upper bound on the max-norm of $M$. 
As S-EVHM is equivalent to the restriction of SLEP
to the case in which $A=\mathds{1}_n$ is fixed and access to $M$ is given by sparse-access oracles,
the lower bound for S-EVHM also applies to SLEP in the sparse-access setting. 
We also derive a lower bound for B-EVHM (Problem~\ref{pro:B-EVHM}),
which is the problem of estimating the 
expectation value of a Hermitian matrix $M$ provided by block access with respect to the state generated by a given 
unitary black box. Our lower bound for B-EVHM is $\Omega(\alpha_M/\epsilon)$,
where $\alpha_M$ is the scaling factor
for the block encoding of $M$. 
As $d_M\beta_M \ge \norm{M}$ and $\alpha_M \geq \|M\|$,
a common $\Omega(\norm{M}/\epsilon)$ 
lower bound holds for estimating the expectation value of a black-box Hermitian matrix $M$
irrespective of the oracle encoding of~$M$.
Our lower bound complements the lower bound on queries to $V$ given by Corollary 4 in Ref.~\cite{HWM0W2021}.
Interestingly, whereas our lower bound for B-EVHM holds under the restriction $V = \mathds{1}_n$, 
our lower bound for S-EVHM is violated under the same restriction.
As previous works did not consider the formulation of SLEP in which block-access to $M$
is given, the optimal scaling with respect to $\alpha_M$ was not known. 
Our lower bound for B-SLEP, which follows from our lower bound for B-EVHM, establishes a linear scaling in $\alpha_M$.

To close, we discuss the implications of our algorithm for B-SLEP. Most importantly, the query complexity of our
algorithm proves that our lower bound for queries to $M$ for B-SLEP is tight up to a constant prefactor. Our algorithm 
for B-SLEP builds on the algorithm for B-EVHM in Ref.~\cite{Ral2020}. 
Our algorithm for B-SLEP also achieves
linear scaling with respect to $\alpha_M$ for queries to $M$, which is optimal, and saturates the constraint on queries to $A$ given by Theorem~\ref{lem:HHLlowerbound} up to logarithmic factors. 
Our algorithm for B-SLEP can be used to compute the expectation value of a smooth function of a Hermitian matrix 
with respect to the solution vector $\bm{x}$ 
by combining our algorithm with existing quantum algorithms for matrix arithmetic \cite{LC19,CGJ19,CSTWZ2019}. 
Finally, our algorithm for S-EVHM shows that our lower bound for the same problem is tight. This algorithm can be 
used as a subroutine along with Lemma~\ref{lem:block encoding sparse access} to construct an algorithm 
for SLEP in the sparse-access setting with optimal queries to $M$, although we do not detail this algorithm in our paper.

\section{Conclusion}
\label{sec:conclusion}
Systems of linear equations arise in nearly all areas of science and engineering. 
The HHL algorithm and its improvements generate a quantum encoding of the solution vector $\bm x$ to a system of
$N$ linear equations $A\bm x = \bm b$
in cost $\polylog(N)$~\cite{HHL09}. This logarithmic dependence on~$N$ can be translated 
into efficient solutions for practical problems
if the value $\bm x^\dagger M \bm x$ can be estimated for a given $M$ in cost $\polylog(N)$. 
However, the cost of this estimation can be prohibitively expensive~\cite{Aaronson2015}. 
To address this caveat, we determine the complexity of the System of Linear Equations Problem (SLEP), 
in which the task is to estimate $\bm x^\dagger M \bm x$ to an additive accuracy $\epsilon$ 
given $A$, $M$, and $\bm b$. We consider the setting in which $M$ is given as a 
block encoding~\cite{Gilyn2019}. Block encoding is a common method 
of specifying matrix input to computational problems,
and it also allows derivation of lower bounds on the query cost.

We analyze the cost of solving SLEP in terms of queries to the block encoding of $M$.
Our main result is that any quantum algorithm for solving SLEP has a cost that 
is of order $\alpha_M/\epsilon$ if $M$ is provided as a block encoding, where $\alpha_M$ is the proportionality factor of the block encoding. 
As $\log(\alpha_M/\epsilon)+1$ digits are sufficient to provide an estimate  
to additive accuracy $\epsilon$, our results imply that
the cost of solving SLEP scales exponentially with the number of digits to which $\bm x^\dagger M \bm x$ 
is estimated. 
We also prove that our bound on the cost is tight 
with respect to $\epsilon$ by constructing a quantum algorithm that saturates our bound.
This bound for SLEP can be used in conjuction with previous hardness results, which
bound the queries to the encoding of~$A$ with respect to~$\epsilon$~\cite{HHL09}. 
To prove our lower bound for SLEP, we first prove a lower bound  
for the problem B-EVHM, in which the task is to estimate the expectation value of a block-encoded matrix $M$ 
 with respect to a given quantum state. To derive a lower bound for B-EVHM, we leverage a known lower bound 
for estimating the mean of a black-box function~\cite{NW99}.
Our algorithm for SLEP relies on a known algorithm for B-EVHM~\cite{Ral2020}. 

Our results rigorously prove that if $M$ is provided by a block encoding and if~$\alpha_M/\epsilon$ is super-polylogarithmic in~$N$, 
then the cost of any quantum algorithm for solving SLEP is also super-polylogarithmic in~$N$. Consequently, quantum algorithms 
that employ a block encoding of $M$
show promise only for those applications for which an accuracy with $\alpha_M/\epsilon \in \polylog(N)$ is sufficient. 
If the scaling of $\alpha_M/\epsilon$ with respect to $N$ is faster than $\polylog(N)$,
then an efficient solution is unlikely unless $M$ is further restricted.
An interesting question in this direction is whether such an efficient solution is possible if constraints on the rank or the eigenvalue distribution of $M$ are known, in addition to restrictions on $A$ and $\bm{b}$.

\acknowledgements
We acknowledge the traditional owners of the land on which this work was undertaken at the University of Calgary: the Treaty 7 First Nations. This project is supported by the Government of Alberta and 
by the Natural Sciences and Engineering Research Council of Canada (NSERC).
A.\ A.\ acknowledges support through the Killam 2020 Postdoctoral Fellowship. This work was partially completed while R. R. Nerem was at the Quantum Algorithms Institute, where he was funded by a Mitacs Accelerate Scholarship.

\appendix

\section{Lower bound for B-EVHM for fixed $V$}
\label{sec:appendix}
In this appendix, we show that the lower bound B-EVHM given by Theorem~\ref{thm:queryB-EVHM} also applies to 
the restriction of B-EVHM in which $V$ is a fixed input, and not a unitary black box. We first adapt a theorem which shows that block access to two matrices $A$ and $B$ can be used to construct block access to $AB$.
\begin{lemma}
[\cite{Gilyn2019}]\label{thm:multiplying block encodings}
Given $n$ qubits,  block access $(\alpha_A, a_A, 0, U_A)$ to $2^n \times 2^n$ Hermitian matrix $A$ and block access  $(\alpha_B, a_B, 0, U_B)$ to $2^n \times 2^n$ Hermitian matrix $B$, then block access
$(\alpha_A\alpha_B, a_A + a_B, 0, U_{AB})$ to $AB$ can be constructed such that $U_{AB}$ uses only a single call to $U_A$ and a single call to $U_B$. 
\end{lemma} 

We now show that Theorem~\ref{thm:queryB-EVHM} applies when to B-EVHM under any restriction that fixes $V$. 
\begin{lemma}
\label{lem:B-EVHMfixedV}
Let $\{V_n \in \mathcal{U}(\mathscr{H}_2^{\otimes n}): \ n \in \mathbb{Z}^+\}$ be a set of unitary operators. 
Then any algorithm $\mathcal{Q}$ that solves all instances of B-EVHM with input $(n,(\alpha_M,a_M,0,U_M),1/\epsilon,V=V_n)$
makes $\Omega(\alpha_M/\epsilon)$ queries to the block encoding of $M$. 
\end{lemma}
\begin{proof}
We first prove that given an algorithm $\mathcal Q$ as described in the Lemma~\ref{lem:B-EVHMfixedV}, any instance of B-EVHM
can be solved by making one query to $\mathcal{Q}$ and without making any additional uses of the block encoding of $M$. 
The desired lower bound then follows immediately from Theorem~\ref{thm:queryB-EVHM}. 
Note that 
\begin{equation}
    U_F = (V_n V^\dagger \otimes \ket{0^{a_M}}\bra{0^{a_M}})U_M (V V_n^\dagger\otimes \ket{0^{a_M}}\bra{0^{a_M}})
\end{equation}
is a $(\alpha_M,a_M,0)$ block encoding of $F = V_n V^\dagger M V V_n^\dagger$. This block encoding can be instructed by Lemma~\ref{thm:multiplying block encodings}. 
Then any instance $(n,(\alpha_M,a_M, 0,U_M),1/\epsilon,V)$  of B-EVHM can be solved by
making one query to $\mathcal{Q}$ with input 
$(n,(\alpha_M,a_M, 0,U_F),1/\epsilon)$. The output $r$ of this algorithm satisfies
\begin{equation}
    \abs{r - \braket{0^n|V_n^\dagger V_n V^\dagger M V V_n^\dagger V_n|0^n}} < 2/3 
\end{equation}
which implies
\begin{equation}
    \abs{r - \braket{0^n|V^\dagger M V |0^n}} < 2/3
\end{equation}
with probability at least $2/3$ as required.
\end{proof}

\section{Algorithm for S-EVHM}
\label{app:sslep}
In this appendix, we construct an algorithm for S-EVHM to show that the query-complexity lower bound 
for S-EVHM given by Corollary.~\ref{col:LB for S-EVHM} is tight.
\begin{thm}
\label{thm:S-EVHM}
Given $n$ qubits, 
sparse access $(d_M, \beta_M,O_{M_{\rm val}}, O_{M_{\rm loc}})$ to a $2^n \times 2^n$ Hermitian matrix $M$, 
an accuracy~$\epsilon > 0$ and a unitary black box $V$, then
an $\epsilon$-additive approximation $\tilde{u}$ to $u := \braket{0 |V^\dagger M V| 0 }$ 
can be returned, with probability at least 2/3,
using $\mathcal  O(d_M\beta_M/\epsilon)$ queries to
sparse-access oracles for $M$, 
$\mathcal  O(d_M\beta_M/\epsilon)$
queries to $V$
as well as 
\begin{equation}
\label{eq:S-EVHM 2qg}
   \mathcal O\left(\frac{d_M \beta_M }{\epsilon} \left(n+\log ^{2.5}\left(\frac{d_M\beta_M}{\epsilon}\right)\right)\right),
\end{equation}
additional 2-qubit gates.
\end{thm}

\begin{proof}
The algorithm is as follows: 
\begin{algorithm}[H]
\begin{algorithmic}[1]
\caption{Algorithm for Sparse-access Expectation Value of an Hermitian Matrix \label{alg:S-EVHM}}
\Require{
$n$ qubits, 
sparse access $(d_M, \beta_M,O_{M_{\rm val}}, O_{M_{\rm loc}})$ to $M$, an accuracy $\epsilon > 0$, and a black box unitary $V$}
\Ensure{
an $\epsilon$-additive approximation $\tilde{u}$ to $u := \braket{0|V^\dagger M V| \psi }$ 
with probability 2/3
}

\hspace{-1.3cm}\ \textbf{Procedure:}
\State Construct block access $(d_M\beta_M,2,\epsilon,U_M)$ to  $M$. 
\Comment{use Lemma \ref{lem:block encoding sparse access}}.
\State Compute $\tilde{u}$ by executing Algorithm \ref{alg:AB-EVHM} with input 
$(n,(d_M\beta_M,2,\epsilon/2,U_M),\epsilon/2,V)$.
\State \Return $\tilde{u}$.
\end{algorithmic}
\end{algorithm}
\begin{proof}[Correctness]
Let $\hat{M} := d_M\beta_M\braket{0^2|U_M|0^2}$, so that $\norm{\hat{M}-M} < \epsilon/2$.
Then Algorithm~\ref{alg:AB-EVHM} guarantees that
\begin{equation}
    \abs{\tilde{u} - \braket{0^n|V^\dagger \hat{M} V|0^n}} < \epsilon/2,
\end{equation}
with probability at least $2/3$. We then have 
\begin{align}
    \abs{\tilde{u} - u} &\le \abs{\tilde{u} - \braket{0^n|V^\dagger \hat{M} V|0^n}} + 
    \abs{\braket{0^n|V^\dagger \hat{M} V|0^n} - \braket{0^n|V^\dagger M V|0^n}} \nonumber\\
    & \le \epsilon/2 + \norm{\hat{M}-M} < \epsilon,
\end{align}
with probability at least 2/3 as required.
\renewcommand{\qedsymbol}{}
\end{proof}

\begin{proof}[Complexity] 
The complexity of queries to sparse-access oracles follows directly from the complexities of Algorithm~\ref{alg:AB-EVHM} with $\alpha_M$ set to $d_M \beta_M$. By Lemma \ref{lem:block encoding sparse access}, the block encoding $U_M$ uses $\mathcal{O}\left(n+\log ^{2.5}\left(d_M\beta_M/\epsilon\right)\right)$ additional 2-qubit gates. Furthermore, $U_M$ is queried $\mathcal O(d_M \beta_M /\epsilon) $ 
times by Algorithm \ref{alg:AB-EVHM}. Therefore, our algorithm for S-EVHM uses 
the number of 2-qubit gates given in Eq. (\ref{eq:S-EVHM 2qg}). 
\end{proof}
\renewcommand{\qedsymbol}{}
\end{proof}
\noindent The queries to the block encoding of $M$ given in Theorem~\ref{thm:B-EVHM} saturates the lower bound for B-EVHM in
Theorem~\ref{thm:queryB-EVHM}. This theorem concludes the proof of the tightness of our lower bound
for B-EVHM.

\bibliography{References3}

\end{document}